\documentclass[11pt]{amsart}

\usepackage[utf8]{inputenc} 
\usepackage[T1]{fontenc}    
\usepackage{amsmath,amsthm,amsfonts,amssymb,latexsym,verbatim,bbm}
\usepackage{url}            
\usepackage{booktabs}       

\usepackage{nicefrac}       
\usepackage{microtype}      
\usepackage{xcolor}         
\usepackage{hyperref}
\allowdisplaybreaks[4]
\usepackage{indentfirst}

\usepackage{graphicx}
\usepackage{subfigure}
\usepackage{float} 

\usepackage{url}
\usepackage{lscape}
\usepackage{booktabs}
\usepackage{multirow}
\usepackage{mathrsfs}
\usepackage{venndiagram}
\usepackage{bm}
\usepackage{comment}
\usepackage{algorithm}
\usepackage{xcolor}
\usepackage{algorithmic}

\usepackage{braket}

\usepackage{subcaption}
\usepackage[shortlabels]{enumitem}
\usepackage{afterpage}
\usepackage{mathtools}
\usepackage{mathrsfs}

\usepackage[capitalize]{cleveref}

\newtheorem{theorem}{Theorem}[section]
\newtheorem{lemma}[theorem]{Lemma}
\newtheorem{proposition}[theorem]{Proposition}

\newtheorem{definition}{Definition}[section]
\newtheorem{remark}{Remark}[section]

\newtheorem{problem}{Problem}[section]

\newcommand{\pp}{\mathbb{P}}

\newcommand{\ie}{\textit{i.e.}}

\newcommand{\one}{\mathbbm{1}}

\newcommand{\image}{\textnormal{Im}}
\newcommand{\etal}{\textit{et al.}}

\newcommand{\Span}{\textnormal{span}}

\newcommand{\NP}{\mathsf{NP}}

\newcommand{\Hom}{\textnormal{Hom}}
\newcommand{\Herm}{\textnormal{Herm}}
\newcommand{\Conv}{\mathop{\scalebox{1.5}{\raisebox{-0.2ex}{$\ast$}}}}
\newcommand{\prob}[1]{\textnormal{\texttt{#1}}}

\newcommand{\Tr}{\textnormal{Tr}}

\usepackage[top=1.3in,bottom=1.3in,left=1.3in,right=1.3in]{geometry}

\title[Topics in Non-local Games]{Topics in Non-local Games: Synchronous Algebras, Algebraic Graph Identities, and Quantum NP-hardness Reductions 
}

\author[E. He]{Entong He$^{1}$} 

\address[1]{Department of Computer Science, The University of Hong Kong, Pokfulam, Hong Kong SAR, China}

\email{entong\_he@connect.hku.hk}

\begin{document}

\begin{abstract}
    We review the correspondence between synchronous games and their associated $*$-algebra. Building upon the work of (Helton \etal, New York J. Math. 2017), we propose results on algebraic and locally commuting graph identities. Based on the noncommutative Nullstellens\"atze (Watts, Helton and Klep, Annales Henri Poincar\'e 2023), we build computational tools that check the non-existence of perfect $C^*$ and algebraic strategies of synchronous games using Gr\"obner basis methods and semidefinite programming. We prove the equivalence between the hereditary and $C^*$ models questioned in (Helton \etal, New York J. Math. 2019). We also extend the quantum-version $\NP$-hardness reduction $\texttt{3-SAT}^* \leq_p \texttt{3-Coloring}^*$ due to (Ji, arXiv 2013) by exhibiting another instance of such reduction $\texttt{3-SAT}^* \leq_p \texttt{Clique}^*$.
\end{abstract}

\maketitle

\tableofcontents

\section{Introduction}
A generic version of a two-player non-local game involved two spatially separated players and a referee. Players attempt to convince the referee of the existence of a certain property by conducting several rounds of interactions with the referee. A Boolean function, called a predicate, evaluates each pair of inputs and outputs to indicate whether the players win or lose in the game. A comprehensive introduction can be found in \cite{Cleve2014}. Generally, when the players share classical randomness, they can win perfectly if and only if the underlying Boolean constraint system (BCS) is satisfiable. However, shared quantum resources generalize the classical strategy space, yielding quantum strategies. There are examples of games where quantum strategies admit non-zero advantage \cite{Bell1964, Mermin1990, Peres1990}. BCS non-local games also motivate new topics in complexity theory. Ji suggested the satisfiability problems with quantum operator assignment and lifted several NP-hardness reductions to their quantum versions in \cite{Ji2013}. In \cite{atserias2017generalized}, it is proved that quantum advantages only exist for certain types of constraint satisfaction problems (CSP).
\par The synchronous game is a special subclass of non-local game, where two players share the same input and output set, and the predicate requires a `consistency check' \cite{helton2023synchronousvaluesgames}. In \cite{BeneWattsNullstellensatze2023}, the theory of noncommutative (NC) real algebraic geometry is developed which bears on the perfectness of synchronous games. Graph homomorphism games are a subclass of synchronous games, and from certain perfect strategies for the homomorphism game, we can abstract new graph identities. A detailed discussion is provided in \cite{Ortiz2016}. Helton, Meyer, Paulsen, and Satriano \cite{helton2019algebras} investigate the existence of various types of operator-valued graph homomorphisms by examining properties of the associated $*$-algebra of the associated synchronous games. They also suggested two special types of algebra, the hereditary and locally commuting algebra. Some preliminary results are presented from the literature, as well as questions for further exploration.
\par This report introduces the background of non-local games, summarizes some important results, outlines implementations of some computational tools, and contributes some new results. We leave some problems open, hoping that subsequent works will fill these gaps.

\section{Preliminaries and Notations}
\subsection{Non-local games and synchronous games} 
A \textbf{two-player finite input-output game} for players Alice ($A$) and Bob ($B$) is specified by a tuple $\mathcal{G} = (I_A, I_B, O_A, O_B, \lambda)$ where $|I_A|, |I_B|, |O_A|, |O_B| < \infty$ and a predicate $\lambda: I_A \times I_B \times O_A \times O_B \to \{0, 1\}$ indicating the win or lose outcome of a round. A \textbf{strategy} for the players is the collection of probabilities $\left\{\pp(a, b | i, j)\right\}_{i, j, a, b}$ where $(i, j) \in I_A \times I_B$ and $(a, b) \in O_A \times O_B$. It can be established as follows: Two players agree on a measure space $(\Omega, \mathcal{F}, \pp)$ and assign measurable functions $f_i$, $g_j$ for Alice and Bob on each index $i \in I_A$, $j \in I_B$ respectively before the game starts. The functions $f_i: \Omega \to O_A, g_j: \Omega \to O_B$ indicate the strategies. A strategy $\{ f_i, g_j \}_{i \in I_A, j \in I_B}$ is said \textbf{perfect} if 
\begin{equation}
\forall (i, j) \in I_A \times I_B, \quad \pp \left\{ \omega \in \Omega:~ \lambda(i, j, f_i(\omega), g_j(\omega)) = 1 \right\} = 1.
\end{equation}
\par A non-local game is a \textbf{synchronous game} provided that the players are sharing the same input and output spaces, \ie, $I_A = I_B = I$ and $O_A = O_B = O$, plus the predicate satisfies the synchronous condition, \ie, for any $i \in I$, $\lambda(i, i, a, b) = \delta_{a, b}$ for any $a, b \in O$.

\subsection{Basics of groups and $*$-algebras}
Let $S$ be an arbitrary set, $\mathcal{F}_S$ is the free group with generating set $S$, \ie, $\mathcal{F}_S = \{ \prod_{i = 1}^{ n \in \mathbb{N} \cup \{0\} } x_{i}:~x_i \in S \cup S^{-1} \}$. A unital algebra $\mathcal{A}$ (\ie, $1_{\mathcal{A}} \in \mathcal{A}$) is a \textbf{$*$-algebra} (over a field $\mathbb{F}$ with an involution) if it is equipped with an \textit{involution} $x \mapsto x^*$, such that
\textbf{(i)} $(1_{\mathcal{A}})^* = 1_{\mathcal{A}}$ and $(a^*)^* = a$ for any $a \in \mathcal{A}$; \textbf{(ii)} for any $a, b \in \mathcal{A}$, $(ab)^* = b^*a^*$; \textbf{(iii)} for any $a, b \in \mathcal{A}$, $\lambda, \mu \in \mathbb{F}$, $(\lambda a + \mu b)^* = \lambda^* a^* + \mu^* b^*$. A $*$-homomorphism $\pi: \mathcal{A} \to \mathcal{B}$ between two $*$-algebras $\mathcal{A}$ and $\mathcal{B}$ is an algebra homomorphism such that $\forall a \in \mathcal{A}$, $\pi(a^*) = \pi(a)^*$. We define the \textbf{free noncommutative $*$-algebra} on a generating set $S$ by $\mathbb{F}^*\left<S\right> := \{ \sum_{i=1}^{n \in \mathbb{N} \cup \{0\} } \lambda_i \prod_{ j = 1 }^{m_i} x_{i, j}: ~m_i \geq 1,~\lambda_i \in \mathbb{F},~x_{i, j} \in S \cup S^* \}$.
Both groups and $*$-algebras can be defined by \textbf{presentations}. We define $\left<S: R\right>$ as the group quotient $\mathcal{F}_{S} / \left<R\right>$ where $R \subseteq \mathcal{F}_S$ and $\left<R\right> \unlhd \mathcal{F}_{S}$, and define $\mathbb{F}^*\left<S: R\right>$ as the quotient of algebra $\mathbb{F}^*\left<S\right> / \left<\left<R\right>\right>$. For two groups $G = \left<S_G:R_G\right>$, $H = \left<S_H:R_H\right>$, the \textbf{group free product} $G * H = \left< S_G \cup S_H : R_G \cup R_H \right>$. Analogously, we can define the \textbf{algebra free product} of two algebras $G = \mathbb{F}^*\left<S_G: R_G\right>$ and $H = \mathbb{F}^*\left<S_H: R_H\right>$ as $G * H = \mathbb{F}^*\left<S_G \cup S_H: R_G \cup R_H \right>$. For two group or algebra elements $a, b$, we use $[a, b] = ab - ba$ to specify their \textbf{commutator}.
\par The \textbf{sum-of-square cone} $SOS_{\mathcal{A}} := \{ \sum_i a_i^* a_i:~a_i \in \mathcal{A} \}$ is a hermitian subset of $\mathcal{A}$. A partial ordering on $\mathcal{A}$ can be induced by setting $x \leq y$ if $x - y \in SOS_{\mathcal{A}}$ for any $x, y \in \mathcal{A}$. A \textbf{state} on $\mathcal{A}$ is a linear functional $\psi: \mathcal{A} \to \mathbb{F}$ $\in \{\mathbb{R}, \mathbb{C}\}$ such that $\psi(1_{\mathcal{A}}) = 1$, $\psi(a^*) = \psi(a)^*$ and $\psi(a^* a) \geq 0$ for all $a \in \mathcal{A}$, and it is called \textit{tracial} if $\forall a, b \in \mathcal{A}$, $\psi(ab) = \psi(ba)$. 
Intuitively, for any state $\psi$ on $\mathcal{A}$, $\psi(SOS_{\mathcal{A}}) \subseteq \mathbb{R}_{\geq 0}$. We denote $\textnormal{Herm}(\mathcal{A}) = \{a \in \mathcal{A}:~a^* = a \}$ as the hermitian elements in $\mathcal{A}$, and $\mathcal{A}_{bdd}$ as the \textbf{$*$-subalgebra of bounded elements} in $\mathcal{A}$, which is defined by
$$
\mathcal{A}_{bdd} := \left\{a \in \mathcal{A}:~\exists B \in \mathbb{R}_{\geq 0},~a^* a \leq R \cdot 1 \right\}.
$$
We say $\mathcal{A}$ is \textit{Archimedean} if $\mathcal{A} = \mathcal{A}_{bdd}$.

\subsection{Synchronous algebras, and a hierarchy of strategies}
\label{sec:AlgebraDefinition}
The key idea of non-local games is that the shared quantum resources allow a greater feasible set of correlations than the classical 
setting. The existence of perfect strategy for synchronous games can be examined by studying their algebraic properties.
Each synchronous game $\mathcal{G} = (I, O, \lambda)$ is associated with a $*$-algbera, called the \textbf{synchronous algebra}, which is the quotient $\mathbb{C}[\mathcal{F}(|I|, |O|)] / \mathcal{I}(\mathcal{G})$, and is denoted as $\mathcal{A}(\mathcal{G})$. Here $\mathcal{F}(|I|, |O|)$ is the free product of $|I|$ unitary groups, each of order $|O|$. Alternatively, the generators $\{x_{i, a} \}_{i \in I, a \in O}$ abstract the algebraic properties of projection-valued measures (PVMs), and the $*$-closed, two-sided ideal $\mathcal{I}(\mathcal{G})$ is generated by the following relations
\begin{equation} \label{eqn:NonlocalGameRelations}
\begin{array}{ll}
    \bullet~x_{i, a}^2 - x_{i, a}, \quad \forall i \in I, a \in O & \bullet~ x_{i, a}^* - x_{i, a}, \quad \forall i \in I, a \in O \\
    \bullet~ \sum_{a \in O} x_{i, a} - 1, \quad \forall i \in I & \bullet~ x_{i, a} x_{j, b}, \quad \forall (i, j, a, b) \in \lambda^{-1}(\{0\}).
\end{array}
\end{equation}
Based on the existence of unital $*$-homomorphism from $\mathcal{A}(\mathcal{G})$ to certain unital $*$-algebras $\mathcal{B}$, we introduce the concept of \textbf{perfect $t$-strategy} (or equivalently, $t$-satisfiability).  Certain choices of $t$ are $\{loc, q, qa, qc, C^*, hered, alg \}$, and we focus only on a few of them in this work. Specifically, we write $\mathcal{A}(\mathcal{G}) \to \mathcal{B}$ if there is a unital $*$-homomorphism $\rho: \mathcal{A}(\mathcal{G}) \to \mathcal{B}$. \cite{Cleve2014, Paulsen_2016} suggest the following definitions:
\begin{align*}
&\text{1. $\mathcal{G}$ has a perfect $loc$-strategy if $\mathcal{A}(\mathcal{G}) \to \mathbb{C}$.}        \\       &\text{2. $\mathcal{G}$ has a perfect $q$-strategy if $\mathcal{A}(\mathcal{G}) \to \mathbb{C}^{d\times d}$, $d \in \mathbb{N}$.}                                                          \\
&\text{3. $\mathcal{G}$ has a perfect $qa$-strategy if $\mathcal{A}(\mathcal{G}) \to  \mathcal{R}^{\mathcal{U}}$}\footnotemark.  \\
&\text{4. $\mathcal{G}$ has a perfect $qc$-strategy if $\mathcal{A}(\mathcal{G}) \to (\mathcal{T}, \tau)$ for a unital $C^*$-algebra $\mathcal{T}$ with a} \\
&\quad \text{~tracial state $\tau$.}   \\
&\text{5. $\mathcal{G}$ has a perfect $C^*$-strategy if $\mathcal{A}(\mathcal{G}) \to \mathcal{B}(\mathcal{H})$ for a non-trivial Hilbert space $\mathcal{H}$.} \\
&\text{6. $\mathcal{G}$ has a perfect $alg$-strategy if $\mathcal{A}(\mathcal{G}) \neq (0)$.}
\end{align*}
\footnotetext{$\mathcal{R}^{\mathcal{U}}$ is an ultrapower of the hyperfinite $\text{II}_{1}$ factor $\mathcal{R}$, which is beyond the scope of our discussion. For detailed preliminaries see \cite{Ozawa2013}.}
The definition of $hered$-strategy is postponed to Section \ref{sec:heredEqualsCstar}. We refer readers to \cite{helton2019algebras, harris2023universalitygraphhomomorphismgames, Ortiz2016} for detailed information on their characteristics. Denote $\mathcal{M}_{t}$ as the set of $t$-models, \ie, the existence of perfect $t$-strategies of all synchronous games, the following hierarchical order holds:
\begin{equation}
\mathcal{M}_{loc} \subsetneq \mathcal{M}_{q} \subseteq \mathcal{M}_{qa} \subseteq \mathcal{M}_{qc} \subsetneq \mathcal{M}_{C^*} \subseteq \mathcal{M}_{hered} \subsetneq \mathcal{M}_{alg}.
\end{equation}
There are certain constructions of synchronous games where the consecutive neighboring models differ. For instance, $\mathcal{M}_{loc} \neq \mathcal{M}_q$ \cite{Mariia2024quantumindependencechromaticnumbers} and $\mathcal{M}_{qc} \neq \mathcal{M}_{C^*}$ \cite{paddock2023satisfiability}. It is proved in \cite{Ozawa2013} that deciding whether $\mathcal{M}_{qa} = \mathcal{M}_{qc}$ is equivalent to Connes embedding conjecture, which is announced false by Ji, Natarajan, Vidick, Wright, and Yuen in \cite{ji2022mipre} as a by-product of the $\prob{MIP}^* = \prob{RE}$ result. We will prove $\mathcal{M}_{C^*} = \mathcal{M}_{hered}$ for synchronous games.

\section{Algebraic and locally commuting clique number}

We are interested in two special graph identities, \ie, the algebraic clique number, and the locally commuting clique number, which are introduced by Helton, Meyer, Paulsen, and Satriano in \cite{helton2019algebras}. They are derived from the following definition of the graph homomorphism game.
\begin{definition}
   Given two graphs $G, H$, the \textbf{graph homomorphism game} is the synchronous game $\mathcal{G} = (I, O, \lambda)$, such that $I = V(G)$, $O = V(H)$, and predicate $\lambda$ satisfies
   $$
   \lambda(a, a, u, v) = \delta_{u, v}, \quad \lambda(a, b, u, v) = \left\{ 
    \begin{aligned}
        &1 &(a \sim_{G} b \land u \sim_H v) \lor (a \not \sim_{G} b) \\
        &0 &\text{otherwise}.
    \end{aligned}
   \right. 
   $$
   Throughout the report, we use the notation $\Hom(G, H)$ to represent the above graph homomorphism game between $G$ and $H$. We denote $G \overset{alg}{\longrightarrow} H$ if $\mathcal{A}(\Hom(G, H)) \neq (0)$. The relation of algebraic graph homomorphism is transitive, in view of the following theorem.

\begin{proposition}
\label{thm:algTransitivity}
    $G \overset{alg}{\longrightarrow} H$ and $H \overset{alg}{\longrightarrow} K$ implies $G \overset{alg}{\longrightarrow} K$ for any graphs $G, H, K$.

\end{proposition}
\begin{proof}
The proof closely resembles the proof in the $lc$-case in Lemma 8.1 of \cite{helton2019algebras}. 
   For any $g \in V(G), k \in V(K)$, we define $x_{g, k} := \sum_{h \in V(H)} x_{g, h} \otimes x_{h, k}$. With this construction, we claim that $\mathbb{C}\left< x_{g, k} \right>$ is non-trivial, since if $x_{g, k} = 0$ for any $(g, k)$, then $0 = \sum_{k \in V(K)} x_{g, k} = \sum_{k \in V(K)} \sum_{h \in V(H)} x_{g, h} \otimes x_{h, k} = \sum_{h \in V(H)} x_{g, h} \otimes \sum_{k \in V(K)} x_{h, k} = 1$, then one of $\mathcal{A}(\Hom(G, H))$ and $\mathcal{A}(\Hom(H, K))$ is forced to be trivial, a contradiction. By
    $$
    \begin{aligned}
        \sum_{k \in V(K)} x_{g, k} &= \sum_{k \in V(K)} \sum_{h \in V(H)} x_{g, h} \otimes x_{h, k} = \sum_{h \in V(H)} x_{g, h} \otimes \sum_{k \in V(K)} x_{h, k} = 1; \\
        x_{g, k}^2 &= \sum_{h, h' \in V(H)} {x_{g, h} x_{g, h'}} \otimes x_{h, k} x_{h', k} = \sum_{h \in V(H)} x_{g, h}^2 \otimes x_{h, k}^2 \\
        &= \sum_{h \in V(H)} x_{g, h} \otimes x_{h, k} = x_{g, k}; \\
        x_{g, k} x_{g', k'} &= \sum_{h, h' \in V(H)} x_{g, h} x_{g', h'} \cdot \one_{g \sim_G g' \land h \sim_H h'} \otimes x_{h, k} x_{h', k'} \cdot \one_{h \sim_H h' \land k \sim_K k'} \\
        &= x_{g, k} x_{g', k'} \cdot \one_{k \sim_K k' \land g \sim_G g'}; \\
        x_{g, k} x_{g, k'} &= \sum_{h, h' \in V(H)} x_{g, h} x_{g, h'} \otimes x_{h, k} x_{h', k'} = \sum_{h \in V(H)} x_{g, h}^2 \otimes x_{h, k} x_{h, k'} = x_{g, k}^2 \cdot \delta_{k, k'},
    \end{aligned}
    $$
    $\mathbb{C}\left<x_{g, k}\right>$ vanishes on $\mathcal{I}(\Hom(G, K))$, resulting in the desired $\mathcal{A}(\Hom(G, K))$.
\end{proof}

  For any graph $G$, we define the \textbf{algebraic clique number} $\omega_{alg}(G)$ and the \textbf{locally commuting clique number} $\omega_{lc}(G)$ by
   $$
   \begin{aligned}
   \omega_{alg}(G) &= \max\left\{ n \in \mathbb{N}~|~K_n \overset{alg}{\longrightarrow} G  \right\}, \\
   \omega_{lc}(G) &= \max \left\{ n \in \mathbb{N}~|~\mathcal{A}(\Hom(K_n, G)) / \left< [x_{i, u}, x_{j, v}]:~u \sim_G v, \forall i, j \in [n] \right> \neq (0)  \right\}.
   \end{aligned}
   $$
\end{definition}
   They are analogs of the algebraic chromatic number $\chi_{alg}(G)$ and the locally commuting chromatic number $\chi_{lc}(G)$ discussed in \cite{helton2019algebras} (these chromatic numbers are defined by replacing $\Hom(K_n, G)$ by $\Hom(G, K_n)$ in the above definitions). We present two results, as the follow-up of Helton \etal's result on algebraic chromatic numbers and some problems they left open.
   \par It is shown in \cite{helton2019algebras} that there is an isomorphism between each game $*$-algebra $\mathbb{C}[\mathcal{F}(|I|, |O|)] / \mathcal{I}(\mathcal{G})$ to a free unital algebra $\mathbb{C}\left< x_{i, a}:i \in I, a \in O \right> / \mathcal{I}(|I|, |O|)$, where in the latter algebra all relations are preserved except for $x_{i, a}^* - x_{i, a}$. Thus, in the latter context, we do not distinguish between these two algebras since their mutual interpretation is natural. Besides, it is proved in the same literature that $\mathcal{A}_{\mathbb{F}}(\mathcal{G}) = \mathcal{A}_{\mathbb{Q}}(\mathcal{G}) \otimes_{\mathbb{Q}} \mathbb{F}$ for any field $\mathbb{F}$ containing $\mathbb{Q}$, which enables we to analyze the algebraic graph identities using a noncommutative analog of the Gr\"obner basis (GB) method. We herein briefly review the definition of GB and its fundamental properties.

\begin{definition}
    Given an NC multivariate polynomial ring $\mathbb{C}\left<x\right>$ with a fixed monomial order, and $\mathcal{I}$ is an ideal of $\mathbb{C}\left<x\right>$. A subset $S \subseteq \mathcal{I}$ is a \textbf{Gr\"obner basis} of $\mathcal{I}$ if $\left<\textnormal{LT}(S)\right> = \mathcal{I}$, where $\textnormal{LT}(S)$ is the leading terms of members of $S$. If $\mathcal{I}$ admits finite generating relations $\{h_i\}_{i=1}^m$, its NC GB can be obtained by applying the NC version of Buchberger algorithm on the initial set $\{h_1, \dots, h_m\}$ (for details see \cite{Mora1986}).
\end{definition}

\begin{theorem} \cite{Mora1986}
    For any polynomial $f \in \mathbb{C}\left<x\right>$, $\mathcal{I} \subseteq \mathbb{C}\left<x\right>$ is a (left) ideal. If $S$ is a GB of $\mathcal{I}$, then the remainder $r$ of $f$ being (left) divided by $S$ is unique (independent of the order of division). We denote the procedure as $f \longrightarrow_{S} r$, where $r = 0$ if and only if $f \in \mathcal{I}$.
\end{theorem}

With the aid of Gr\"obner basis, the following result is shown:
\begin{theorem}
    \cite{helton2019algebras} Every graph is algebraically $4$-colorable.
\end{theorem}
The proof of the above theorem is computer-assisted, by examining the explicit expression of the NC GB of the ideal and checking whether $1 \in \mathcal{I}(\Hom(G, K_4))$. By a similar argument, we can work out the general expression for the NC GB of the ideal $\mathcal{I}(\Hom(K_n, K_m))$ and show a similar argument for the algebraic clique number.

\begin{theorem}
\label{thm:algCliqueNum}
    Let $K_m$ be the complete graph on $m$ vertices. Then $\omega_{alg}(K_m) = m$ for $m \in \{1, 2, 3\}$, and $\omega_{alg}(K_m) \geq n$ for any $n \in \mathbb{N}$ and $m \geq 4$.
\end{theorem}
\begin{proof}
    We first prove the second assertion. The result is immediate when $n \leq m$, by the inequality $\omega_{alg}(K_{m}) \geq \omega(K_m) = m \geq n$. It suffices to show that the assertion holds for $n > m \geq 4$. Suppose we assign natural indices to vertices of $K_n$ and $K_m$, specified by $[n] = \{1, \dots, n \}$ and $[m] = \{1, \dots, m\}$. Under the natural graded lexicographic monomial ordering
    $$
    \forall i, j \in [n],~u, v \in [m], \quad x_{i, u} \prec x_{j, v} \quad  \Longleftrightarrow \quad i < j \lor (i = j \land u < v), 
    $$
    the Gr\"obner basis of $\mathcal{I}(\Hom(K_n, K_m))$ consists of the following relations
    \begin{equation}
    \begin{aligned}
        &\sum_{v=1}^{m} x_{i, v} - 1, \quad \forall i \in [n]; \quad x_{i, v}^2 - x_{i, v}, \quad \forall i \in [n],~v \in [m-1]; \\
        &x_{i, v} x_{j, v}, \quad \forall i \neq j \in [n],~v \in [m-1]; \quad x_{i, u} x_{i, v}, \quad \forall i \in [n],~u \neq v \in [m-1]; \\
        &1 - \sum_{u=1}^{m-1} \left(x_{i, u} + x_{j, u}\right) + \sum_{ \substack{ u \neq v \in [m-1] } } \left(x_{i, u} x_{j, v} + x_{j, u} x_{i, v}\right), \quad \forall i \neq j \in [n]; \\
        &1 - \sum_{u=1}^{m-1} \left(x_{i, u} + x_{j, u} \right) + \sum_{ \substack{ u \neq v \in [m-1] } } \left(x_{i, u} x_{j, v} + x_{j, u} x_{i, v}\right) - \sum_{ \substack{ u, v, w \in [m-1] \\ u, v, w \text{ distinct} \\ w \neq m - 2 } } x_{i, u} x_{j, v} x_{i, w} \\
        &- \sum_{\substack{ u \in [m-3] \\ v \in [m-1] \setminus \{u\} }} x_{i, u} x_{j, v} x_{i, u} + \sum_{u \in [m-3]} x_{i, m-1} x_{j, u} x_{i, m-2}, \quad \forall i \neq j \in [n]; \\
        &2 - \sum_{u=1}^{m-1} \left(c_{i, u} x_{i, u} + c_{j, u} x_{j, u} + c_{k, u} x_{k, u} \right) \\
        &+ \sum_{u \neq v \in [m-1]} \left[(\one_{c_{i, u} = c_{j, u} = 2} + 1) \cdot x_{i, u} x_{j, u} + (\one_{c_{i, u} = c_{k, u} = 2} + 1) \cdot x_{i, u} x_{k, u}  \right] \\
        &- \sum_{ \substack{ u, v, w \in [m-1] \\ u, v, w \text{ distinct} \\ w \neq m - 2 } } x_{i, u} x_{j, v} x_{k, w} - \sum_{u \neq v \in [m-1]} x_{i, u} x_{j, v} x_{k, u} + \sum_{u \in [m-3]} x_{i, m-1} x_{j, u} x_{k, m-2}, \\
        &\forall \text{ distinct } i, j, k \in [n], \quad \forall t \in [n], \ell \in [m-1],~c_{t, \ell} \in \{1, 2\}, \\
        &\forall \ell, \ell' \in [m-3],\sum_{t \in \{i, j, k\} } \one_{c_{t, \ell} = 2} = 2,~c_{t, \ell} = c_{t, \ell'}; \\
        &\forall \ell \in \{m-1, m-2\},\sum_{t \in \{i, j, k\}} \one_{c_{t, \ell} = 2} = 1.
    \end{aligned}
    \end{equation}
   It follows that $1 \not \in \mathcal{I}(\Hom(K_n, K_m))$ when $m \geq 4$. Meanwhile, computation shows that $K_4$ is the \textit{minimal} graph satisfying $\omega_{alg}(G) = \infty$, and $\mathcal{I}(\Hom(K_{|V(G)|+1}, G))$ are indeed trivial for any proper subgraph $G \subset K_4$. Take $G = K_m$, inequality $m = \omega(K_m) = \omega_{loc}(K_m) \leq \omega_{alg}(K_m) < m+1$ indicates that $\omega_{alg}(K_m) = m$ for $m \in \{1, 2, 3\}$.
\end{proof}

\begin{remark}
    The transitivity of algebraic homomorphism shown in Proposition \ref{thm:algTransitivity} allows us to extend the result in Theorem \ref{thm:algCliqueNum}. If a graph $G$ admits a $4$-clique, then $K_n \overset{alg}{\longrightarrow} G$ for any $n \in \mathbb{N}$ since $K_n \overset{alg}{\longrightarrow} K_4$ and $K_4  \overset{alg}{\longrightarrow} G$ (by $K_4 \longrightarrow G)$. Consequently, it forces $\omega_{alg}(G) = \infty$ as well.
\end{remark}

\begin{theorem}
    $\omega_{C^*}(K_4) = \omega(K_4) = 4$.
\end{theorem}

\begin{proof}
    It suffices to prove $\alpha_{C^*}(K_4) < 5$. By Theorem \ref{thm:NCPsatzSyncGame}, the $5$-clique game on graph $K_4$ admits a perfect $C^*$-strategy if there exists a positive semidefinite matrix $S$ with degree bound $d \in \mathbb{N}$ such that $1 + [x]_d^* S [x]_d \in \mathcal{I}(\Hom(K_5, K_4))$. With the aid of an SDP, we obtained a refutation when $d = 1$ (see document \texttt{Hom(K5, K4)CstarRefutation.txt}\footnote{Specifically, the document contains the value of $\{f_i\}, \{g_i\}$ under basis $[x]_d$ and matrix $S$ such that $1 + [x]_d^* S [x]_d = \sum_{i} f_i h_i + \sum_i h_i^* g_i$, where $h_i$ are the generating relations of $\mathcal{I}(\Hom(K_5, K_4))$.} for their expressions). Thus, from $4 = \omega(K_4) \leq \omega_{C^*}(K_4) < 5$ the result follows. The result yields another strict separation between $C^*$-satisfiability and algebraic satisfiability.
\end{proof}

Lov\'asz introduced the theta function $\vartheta(\cdot)$ \cite{Lovasz}, which bounds the Shannon capacity of a graph. The Lov\'asz sandwich theorem says that for any graph $G$, 
$
\omega(G) \leq \vartheta(\overline{G}) \leq \chi(G)
$. Ortiz and Paulsen \cite{Ortiz2016} strengthened this inequality to $\omega_{C^*}(G) \leq \vartheta(\overline{G}) \leq \chi_{C^*}(G)$. It is questioned in \cite{helton2019algebras} that whether $\omega_{t}(G) \leq \vartheta(\overline{G}) \leq \chi_{t}(G)$ holds for $t \in \{ lc, hered \}$. We propose a method to prove inequality $\omega_{lc}(G) \leq \vartheta(\overline{G})$ for a family of graphs, and postpone the proof of sandwich theorem when $t = hered$ to Section \ref{sec:heredEqualsCstar}. For preciseness, we use $\mathcal{A}_{lc}(\Hom(G, H))$ to denote the locally commuting graph homomorphism game algebra with ideal $\mathcal{I}_{lc}(\Hom(G, H))$ being $\mathcal{I}(\Hom(G, H))$ combined with the locally commuting conditions.

\begin{lemma}
        \label{lemma:neighborIsomorphism}
        \cite{helton2019algebras} For a graph $G$ and $n \geq 2$, for any $(n-1)$-clique $S$ in $G$, and $N_S = \left\{u \in V(G):~\forall v \in S,~u \sim_G v\right\}$ being the fully-connected neighborhood of $S$, then
        \begin{equation}
        \mathcal{A}_{lc}(\Hom(K_n, G)) \cong \bigoplus_{ \substack{ S \subseteq G \\ S \text{ is a $(n-1)$-clique} } } \mathbb{C}^{|N_S|(n-1)!}.
        \end{equation}
    \end{lemma}

\begin{lemma}
\label{lemma:isoSubgraph}
        For any graph $G$, $n \in \mathbb{N}$, and $Q \subseteq G$ is the subgraph obtained by excluding all vertices and edges that are not in any $n$-clique of $G$. Then $\mathcal{A}_{lc}(K_n, G) \cong \mathcal{A}_{lc}(K_n, Q)$.
    \end{lemma}

    \begin{proof}
        It suffices to consider those edges with endpoints in $n$-cliques, but the edge itself is not included in any $n$-cliques. Suppose $e = \{u, v\}$ is such an edge in a $2 \leq m < n$ clique (since the edge itself is a $2$-clique). Employing Lemma \ref{lemma:neighborIsomorphism}, we have
        $$
        \begin{aligned}
        &\mathcal{A}_{lc}(K_n, G) \cong \bigoplus_{\substack{S \subseteq G \\ S \text{ an $(n-1)$-clique}}} \mathbb{C}^{|N_S|(n-1)!} \\
        &\cong \bigoplus_{\substack{S \subseteq G, u \in N_S, v \not \in N_S \\ S \text{ an $(n-1)$-clique}}} \mathbb{C}^{|N_S|(n-1)!} \oplus \bigoplus_{\substack{S \subseteq G, v \in N_S, u \not \in N_S \\ S \text{ an $(n-1)$-clique}}} \mathbb{C}^{|N_S|(n-1)!} \oplus \bigoplus_{\substack{S \subseteq G, u, v \not \in N_S \\ S \text{ an $(n-1)$-clique}}} \mathbb{C}^{|N_S|(n-1)!},
        \end{aligned}
        $$
        while the case $u, v \in N_S$ is excluded by the assumption. Above algebra is equivalent to $\mathcal{A}_{lc}(K_n, G \setminus \{e\})$, by the fact that amid the choices for $(n-1)$-cliques in $G \setminus \{e\}$, exactly one of the three cases holds. Iteratively deleting these edges and vertices by composing the previous results gives rise to the desired subgraph $Q$.
    \end{proof}

\begin{lemma}
    \label{lemma:sumToIdentity}
        In the above subgraph $Q$ and a vertex $v$ connecting all the remaining vertices, if $Q$ is connected, it holds that $\sum_{i \in [n]} x_{i, v} = 1$.
    \end{lemma}
    \begin{proof}
        We notice that the product term $\prod_{j=1}^m \sum_{w} x_{j, w} \neq 0$ if and only if there is a $m$-clique among the vertices involved in the sum over $w$. Since $v$ is in every $n$-clique of $Q$ (else we will have an $(n+1)$-clique $\{v\} \cup S$ for an $n$-clique $S \subseteq Q$). By the locally commuting conditions,
        $$
        \prod_{i=1}^n \left(x_{i, v} + \sum_{ w \in Q \setminus \{v\} } x_{i, w} \right) = \sum_{ \substack{ S \subseteq Q \\ S \text{ an $n$-clique }  }  } \prod_{ \substack{i \in [n] \\ u_i \in S } } x_{i, u_i} = \sum_{i=1}^n x_{i, v} \sum_{ \substack{ S \subseteq Q \setminus\{v\} \\ S \text{ an $(n-1)$-clique }  }  } \prod_{ \substack{ j \in [n] \setminus \{i\} \\ w_j \in S }} x_{j, w_j}. 
        $$
        Here the choices of $u_i$ over $S$ are mutually distinct, \ie, if $S = \{v_1, \dots, v_n\}$, then $u_i = v_{\sigma(i)}$ where $\sigma$ is a permutation of the indices $[n]$.
        Notice that for any $i \in [n]$,
        $$
        \begin{aligned}
        &x_{i, v} \prod_{j \in [n] \setminus \{i\}} \left(x_{j, v} + \sum_{ w \in Q \setminus \{v\} } x_{j, w} \right) = x_{i, v} \sum_{ \substack{ S \subseteq Q \\ S \text{ an $(n-1)$-clique} } } \prod_{\substack{ j \in [n] \setminus \{i\} \\ w_{j} \in S }  } x_{j, w_j} \\
        &= x_{i, v} \cdot \left[ \sum_{ \substack{ S \subseteq Q, v \not \in S \\ S \text{ an $(n-1)$-clique}} } \prod_{\substack{ j \in [n] \setminus \{i\} \\ w_{j} \in S }  } x_{j, w_j} + \sum_{ \substack{ S \subseteq Q, v \in S \\ S \text{ an $(n-1)$-clique} } } \prod_{\substack{ j \in [n] \setminus \{i\} \\ w_{j} \in S }  } x_{j, w_j}  \right] \\
        &= x_{i, v} \cdot \sum_{ \substack{ S \subseteq Q \setminus\{v\} \\ S \text{ an $(n-1)$-clique } }  } \prod_{ \substack{ j \in [n] \setminus \{i\} \\ w_j \in S }} x_{j, w_j}.
        \end{aligned}
        $$
        The second term in the third equality vanishes since $x_{i, v} x_{j, w_j}$ evaluates to $0$ whenever $w_j = v$. It follows that summing over $i$ in the above expression coincides with the previous equation, \ie, 
        \begin{equation}
        \sum_{i=1}^n x_{i, v} \cdot \underbrace{\prod_{j \in [n] \setminus \{i\}} \left(x_{j, v} + \sum_{ w \in Q \setminus \{v\} } x_{j, w} \right)}_{= 1^{n-1} = 1} = \underbrace{\prod_{i=1}^n \left(x_{i, v} + \sum_{ w \in Q \setminus \{v\} } x_{i, w} \right)}_{=1^n = 1},
        \end{equation}
        which directly indicates that $\sum_{i \in [n]} x_{i, v} = 1$. 
    \end{proof}

\begin{theorem}
\label{thm:generatingSpace}
    Let $\{ x_{i, v} \}_{i \in [n], v \in V(G)}$ be the generators of $\mathcal{A}_{lc}(K_n, G)$, then
    \begin{equation} \label{eqn:generatingRelations}
    \Span \left\{ \sum_{i=1}^n x_{i, v} + \mathcal{I}_{lc}(\Hom(K_n, G)) :~v \in V(G)  \right\} = \Span\left\{ e_v:~v \in V(G) \right\} / \mathcal{W},
    \end{equation}
    where $e_v$ is a basis element corresponding to each $v \in V(G)$. Denote $Q \subseteq G$ as the subgraph of $G$ obtained by excluding all the edges and vertices that are not in any $n$-cliques of $G$. Define the neighborhood set $N_S = \{w \in G:~\exists v \in S,~w \sim_G v\}$, the subspace $\mathcal{W}$ is spanned by the following linear relations
    $$
    \begin{aligned}
    &\left\{ e_v:~\forall \text{$(n-1)$-clique $S \subseteq G$, $v \not \in N_S$} \right\} \cup \left\{ e_v - 1:~v \sim_{Q} w,~\forall w \in V(Q)  \setminus \{v\}  \right\} \\
    \cup &\left\{ \sum_{v \in H} e_v - m:~H \subseteq V(Q) \text{ being minimal},~\forall \text{$n$-clique $C \subseteq Q$}, |C \cap H| = m \right\}.
    \end{aligned}
    $$
    For those vertex subset $H$ being not `minimal', it admits a decomposition $H = H_1 \sqcup H_2$ such that $\sum_{v \in H_1} e_v = m_1$, $\sum_{v \in H_2} e_v = m_2$ and it follows that $m = m_1 + m_2$. Thus, we are interested in the minimal $H$s as they induces those linear relations that are building blocks of the subspace.
\end{theorem}

\begin{proof}
    To prove the first relation, we first examine the case when $v \not \in N_S$ for any $(n-1)$-clique $S$ in $G$. The trivial case is when $v$ is isolated, then $\forall u \in G \setminus \{v\},~i, j \in [n]$, $x_{i, u} x_{j, v} = 0$. By $x_{j, v}^2 = x_{j, v}$ and $x_{i, v}x_{j, v} = \delta_{ij} x_{i, v}^2$, this implies that $\mathcal{A}_{lc}(K_n, G) = \mathcal{A}_{lc}(K_{n-1}, G \setminus \{v\}) \oplus \bigoplus_{j=1}^n \mathcal{A}_{lc}(K_{n-1}, G \setminus \{v\}) x_{j, v} \cong \mathcal{A}_{lc}(K_{n-1}, G \setminus \{v\})$ by the locally commuting condition. Thus every $x_{i, v}$ vanishes and $\sum_{i \in [n]} x_{i, v} = \sum_{i \in [n]} 0 = 0$ when $n$ is enumerable.
    
    Suppose $S \subseteq G$ is the maximal $m$-clique in $G$ such that $m < n - 1$ and $v \in N_S$. Without loss of generality, we assume that $G$ is connected, and for disconnected graphs $G$ it suffices to consider the analog of the proof on each subalgebra of $\Conv_{j} \mathcal{A}_{lc}(\Hom(K_n, \mathcal{C}_j))$, where $\mathcal{C}_j$ are the connected components of $G$. Denote $H_v := \left[V(G) \setminus \left(S \cup \{v\}\right)\right] \cap N_v$, by Lemma \ref{lemma:neighborIsomorphism}, we have
    $$
    \mathcal{A}_{lc}(K_{n}, G) \cong A_{lc}\left(K_{n-1}, S \sqcup H_v\right) \oplus \bigoplus_{y \neq v} \mathcal{A}_{lc}(K_{n-1}, N_y).
    $$
    We claim that in $S \sqcup H_v$ there is no $(n-1)$-clique, since by a contradictory argument, if there is, then $S \sqcup H_v$ will be the maximal clique where $v \in N_{S \sqcup H_v}$ and $|S \sqcup H_v| = n - 1 > m$, contradicting the minimality assumption. Thus $A_{lc}\left(K_{n-1}, S \sqcup H_v\right) = (0)$. As for the second subalgebra, 
    $$
    \begin{aligned}
    \bigoplus_{y \neq v} \mathcal{A}_{lc}(K_{n-1}, N_y) &= \bigoplus_{y \in S \sqcup H_v} \mathcal{A}_{lc}(K_{n-1}, N_y) \oplus \bigoplus_{y \in V(G) \setminus (N_v \cup \{v\})} \mathcal{A}_{lc}(K_{n-1}, N_y) \\
    &\cong \bigoplus_{y \in V(G) \setminus (N_v \cup \{v\})} \mathcal{A}_{lc}(K_{n-1}, N_y),
    \end{aligned}
    $$
    where the second isomorphism arises from the fact that there is no $(n-2)$-clique in $N_y$ for any $y \in S \sqcup H_v$, since if there is, together with $\{y, v\}$ they will form a $n$-clique. Thus, $\mathcal{A}_{lc}(K_{n}, G) \cong \bigoplus_{y \in V(G) \setminus (N_v \cup \{v\})} \mathcal{A}_{lc}(K_{n-1}, N_y)$, and $x_{i, v}$ annihilates since $v \not \in N_y$ for any $y \in V(G) \setminus (N_v \cup \{v\})$. It follows again that $\sum_{i \in [n]} x_{i, v} = 0$. 
    
    
    The second linear relation can be obtained from Lemma \ref{lemma:isoSubgraph} and \ref{lemma:sumToIdentity}, and the last relation can be analogously derived. Suppose $H = \{v_1, \dots, v_{|H|}\}$, each $n$-clique in $Q$ has exactly $m \geq 1$ members in $H$. Then by an analogous statement, suppose the $n$-cliques in $Q$ are listed as $\{C_{\ell}\}_{\ell=1}^{k}$, and each $v \in H$ falls into cliques $\{C_{\ell}\}_{\ell \in I_v}$, where $I_v \subseteq [k]$. We have 
    $$
    \begin{aligned}
    \sum_{i \in [n]} \sum_{v \in H} x_{i, v}  \prod_{j \in [n] \setminus \{i\} } \left(\sum_{u \in Q } x_{j, u} \right) &= \sum_{i \in [n]} \sum_{v \in H} x_{i, v}  \sum_{\ell \in I_v } \sum_{w_j \in C_{\ell} \setminus \{v\}} \prod_{ j \in [n] \setminus \{i\} } x_{j, w_j} \\
    &= \sum_{v \in H} \sum_{\ell \in I_v} \sum_{ w_j \in C_{\ell} } \prod_{j \in [n]} x_{j, w_j} = m \cdot \underbrace{\sum_{\ell = 1}^k \sum_{w_j \in C_{\ell}} \prod_{j \in [n]} x_{j, w_j}}_{ = \prod_{i=1}^n \left(\sum_{v \in Q} x_{i, v} \right) = 1  } = m.
    \end{aligned}
    $$
    where the second last equality is due to $\biguplus_{v \in H, \ell \in I_v} V(C_{\ell}) = \biguplus_{i\in[m], \ell \in [k]} V(C_{\ell})$, \ie, traversing vertices $v \in H$ and summing the exponent-free product corresponds to the $n$-cliques that $v$ falls in results in an $m$-copy of all the products of $n$-cliques in $Q$. Here $\biguplus$ stands for the multiset.
    \par The case for $Q$ being disconnected is immediate, since if there are connected components $\mathcal{C}_1, \mathcal{C}_2, \dots, \mathcal{C}_r$ and such set $H$ yields the form $H = \bigsqcup_{j=1}^r H_j$ where $H_j$ is subset of each component $\mathcal{C}_j$. Each $n$-clique of $Q$ falls into exactly one of $\{\mathcal{C}_j\}_{j \in [r]}$, and its vertices will appear only in one of the vertex sets $\{H_j\}_{j\in[r]}$. By an analogous statement we have $\sum_{v \in H} \sum_{i \in [n]} x_{i, v} = \sum_{j \in [t]} \sum_{v \in H_j} \sum_{i \in [n]} x_{i, v} = m$, and we omit the details here.
    \par Finally, by composing all the statements above, we take mapping $\sum_{i \in [n]} x_{i, v} \mapsto e_v$ for each $v \in V(G)$ and the result follows immediately.
\end{proof}

\begin{remark}
    The relation $\{e_v - 1:~\forall u \in V(Q) \setminus \{v\},~u \sim_Q v  \}$ in (\ref{eqn:generatingRelations}) can be combined with the latter one. Since by taking $H = \{v\}$, $C \cap H = \{v\}$ holds for any $n$-clique $C$ in $Q$, whence $m = 1$.
\end{remark}

Theorem \ref{thm:generatingSpace} allows us to mimic Sobchuk's constructive proof of $\alpha_{q}(G) \leq \vartheta(G)$ in  \cite{Mariia2024quantumindependencechromaticnumbers} (section 6.4), and give a proof of $\omega_{lc}(G) \leq \vartheta(\overline{G})$ for a certain family of graph $G$. Prior to our construction, we present a preliminary result on the Lov\'asz theta function.

\begin{theorem}
\label{thm:Lovasz}
    For any graph $G$, denote $\bm{1}$ as the all-one vector in $\mathbb{C}^{|V(G)|}$, then
    \begin{equation} \label{eqn:LovaszSDP}
    \vartheta(G) = \left\{
    \begin{array}{cl}
        \sup\limits_{\mathscr{S}} & \left< \mathscr{S}, \bm{1} \bm{1}^T \right> \\
        \textnormal{s.t.} & \mathscr{S}_{u, v} = 0, \quad \forall u, v \in V(G),~u \sim_G v; \\
        &\textnormal{Tr}(\mathscr{S}) = 1; \quad \mathscr{S} \in \mathbb{S}_{+}^{|V(G)|}\footnotemark;
    \end{array} \right\}.
    \end{equation}
    \footnotetext{we use $\mathbb{S}_{+}^{\ell}$ to specify the cone of symmetric positive semidefinite matrices with dimension $\ell \times \ell$.}
\end{theorem}

\begin{proposition}
\label{prop:lovasz}
    For a graph $G$ such that for any such vertex subset $H$ we have $|C \cap H| = 1$ ($m = 1$) for any $n$-clique of $C$ of subgraph $Q$, there exists a feasible decision variable $\mathscr{S}$ in the SDP (\ref{eqn:LovaszSDP}) associated with $\vartheta(\overline{G})$, such that $ \left<\mathscr{S}, \bm{1}\bm{1}^T\right> / \textnormal{Tr}(\mathscr{S}) = n$.
\end{proposition}
\begin{proof}
    Suppose $H = \{H_1, \dots, H_r\}$ is the set of all such vertex subsets. Take $H_{\max} = \mathop{\arg\max}_{i \in [r]} |H_i|$, where $H_{\max} = \{v_{k_1}, \dots, v_{k_h}\}$, we define a linear map $\tau: \{1\} \cup \{e_v:v \in V(G)\} \to \mathbb{C}^h$ by
    $$
    \tau(1) = \bm{1}^h, \quad \tau(e_{v_{k_i}}) = e_{i} \in \mathbb{C}^h: \quad [e_{i}]_j = \delta_{ij}.
    $$
     As for any vertex subset $H \neq H_{\max}$, we keep $\tau$ to be consistent on $H_{\max} \cap H$. For vertices in $H \setminus H_{\max} := \{v_{t_1}, \dots, v_{t_p}\}$, the index set $[h]$ admits a decomposition $[h] = \bigsqcup_{j=1}^p I_j$, where $I_j$ is the neighboring vertices of $v_{t_j}$ in $H_{\max}$, \ie, $I_j = N_{v_{t_j}} \cap H_{\max}$. Note that $\{I_j\}_{j=1}^p$ are disjoint sets, since if $I_{j_1} \cap I_{j_2} \neq \varnothing$ for some $j_1, j_2 \in [p]$ indicates that $v_{t_{j_1}}$ and $v_{t_{j_2}}$ are in the same $n$-clique $C$, making $|C \cap H| \geq 2$, a contradiction. Then we define $\tau(e_{v_{t_j}}) = \sum_{i \in I_j} \tau(e_{v_{k_i}})$ for each $j \in [p]$. Finally, for those vertices $w \in V(G) \setminus V(Q)$, we set $\tau(e_w) = \bm{0}^h$.
    \par If we set up a matrix $\mathscr{S} \in \mathbb{C}^{|V(Q)|}$, such that $\mathscr{S}_{u, v} = \left<\tau(e_u), \tau(e_v)\right>$. Then if $u \not \sim_Q v$, they must fall into different $n$-cliques in $Q$, and by our construction, $u$ and $v$ are neighbors of two vertices in $H_{\max}$, and thus $\tau(e_u)$ and $\tau(e_v)$ admit no common non-zero entries in any dimensions. Then it follows immediately that $\mathscr{S}_{u, v} = 0$ for all $u \not \sim_Q v$.

    Since the image of $\tau$ are $\{0, 1\}$-vectors, it holds that $\|\tau(e_v)\|^2 = \left<\tau(e_v), \bm{1}^h\right>$. Thus,
    $$
    \begin{aligned}
    \textnormal{Tr}(\mathscr{S}) &= \sum_{v \in V(G)} \left<\tau(e_v), \tau(e_v)\right> = \sum_{v \in V(Q)} \left<\tau(e_v), \tau(e_v)\right> = \sum_{v \in V(Q)} \left<\tau(e_v), \bm{1}^h\right> \\
    &= \left<\tau\left(\sum_{v \in V(Q)} e_v \right), \bm{1}^h\right> = \left<\tau(n \cdot 1), \bm{1}^h\right> = \left<n \cdot \bm{1}^h, \bm{1}^h\right> = nh,
    \end{aligned}
    $$
    and
    $$
    \begin{aligned}
    \left< \mathscr{S}, \bm{1}^h {\bm{1}^h}^T \right> &= \sum_{u \in V(G)} \sum_{v \in V(G)} \left<\tau(e_u), \tau(e_v)\right> = \left< \tau\left(\sum_{u \in V(G)} e_u \right), \tau\left(\sum_{v \in V(G)} e_v \right) \right> \\
    &= \left<\tau(n \cdot 1), \tau(n \cdot 1)\right> = \left< n \cdot \bm{1}^h, n \cdot \bm{1}^h \right> = n^2 h.
    \end{aligned}
    $$
    From the construction, $\mathscr{S}$ is a Gram matrix, and is thus in the cone $\mathbb{S}_{+}^{|V(G)|}$. Finally, $\left<\mathscr{S}, \bm{1}\bm{1}^T\right> / \textnormal{Tr}(\mathscr{S}) = n^2h / nh = n $.
\end{proof}

By Proposition \ref{prop:lovasz}, whenever the quotient algebra $\mathcal{A}_{lc}(\Hom(K_n, G))$ is non-trivial, we can construct the linear map $\tau$ on $\{1\} \cup \{x_{i, v}\}_{i \in [n], v \in V(G)}$, then further build a feasible SDP decision variable from the image of $\tau$. By Theorem \ref{thm:Lovasz}, take $n = \omega_{lc}(G)$, it follows that $\vartheta(\overline{G}) \geq n = \omega_{lc}(G)$.


\begin{problem}
    Can we generalize the above construction to the graphs on which $m \geq 2$? Examples include the cycle graph $C_k$ with $k \in 2 \mathbb{N} + 1$ (or graphs containing $C_k$ as an induced subgraph and has clique number $2$), the only subset $H$ is $H = V(C_k)$ with $m = 2$.
\end{problem}

\begin{remark}
    $\mathcal{A}_{lc}(\Hom(K_n, G))$ is in fact Abelian, since for any $u, v \in V(G)$ and $i, j \in [n]$, $[x_{i, u}, x_{j, v}] = \one_{i = j, u = v} \cdot [x_{i, u}, x_{i, u}] + \one_{i \neq j, u = v} \cdot [0, 0] + \one_{i \neq j, u \not \sim_G v} \cdot [0, 0] + \one_{i \neq j, u \sim_G v} \cdot [x_{i, u}, x_{j, v}] = 0$.
\end{remark}

\section{Noncommutative Nullstellens\"atze and its algorithmic implementation}
In \cite{BeneWattsNullstellensatze2023}, Watts, Helton, and Klep discuss the algebraic characterization of general non-local games. They utilized a result in the NC algebraic geometry, known as the NC Nullstellens\"atz, to transform the decision of the existence of perfect $C^*$-strategy and perfect commuting operator ($qc$-) strategy into the decision of noncommutative semialgebraic set membership. Their main result in the paper is summarized as follows:

\begin{theorem} \cite{BeneWattsNullstellensatze2023}
\label{thm:NCPsatzSyncGame}
    Let $\mathcal{A}$ be a $*$-algebra, and $\mathcal{L} \subseteq \mathcal{A}$ is a left ideal. Define the set of positive and tracial terms in $\mathcal{A}$ as
    $$\widetilde{SOS}_{\mathcal{A}} := \left\{ a \in \mathcal{A}:~\exists~ b \in SOS_{\mathcal{A}},~a = b + \sum_{i} [x_i, y_i],~x_i, y_i \in \mathcal{A} \right\}, $$ then
    \begin{itemize}
        \item If $SOS_{\mathcal{A}}$ is Archimedean, or equivalently, $\mathcal{A}$ is Archimedean, the following statements are equivalent:
        \begin{itemize}
            \item There exists a $*$-representation $\mathcal{A} \to \mathcal{B}(\mathcal{H})$ and a non-trivial state $\psi \in \mathcal{H}$, such that $\pi(\mathcal{L}) \psi = \{0\}$.
            \item $-1 \not \in SOS_{\mathcal{A}} + \mathcal{L} + \mathcal{L}^*$.
        \end{itemize}

        \item If $\widetilde{SOS}_{\mathcal{A}}$ is Archimedean, the following statements are equivalent:
        \begin{itemize}
            \item There exists a $*$-representation $\mathcal{A} \to \mathcal{B}(\mathcal{H})$ and a non-trivial \textbf{tracial} state $\psi \in \mathcal{H}$, such that $\pi(\mathcal{L}) \psi = \{0\}$.
            \item $-1 \not \in \widetilde{SOS}_{\mathcal{A}} + \mathcal{L} + \mathcal{L}^*$.
        \end{itemize}
    \end{itemize}
\end{theorem}
The above theorem allows us to transform the decision problem for the non-existence of perfect strategy of certain non-local games into a search problem for specific polynomials in noncommutative $*$-algebra. Certain toolkits can be applied to conduct the computation, and the hermitian sum-of-square terms are frequently encoded by positive-semidefinite variables. The searching procedure thus turns into a semidefinite programming (SDP) \cite{BoydSDP}. We start by presenting a theorem on deciding whether a semialgebraic set (\ie, a set defined by multivariate polynomial equalities and inequalities) is empty, and provide detailed steps of its reformulation into an SDP.

\begin{theorem} \cite{PabloPhD}
\label{thm:psatz}
    Let $\{f_i\}_{i=1}^r$, $\{g_j\}_{j=1}^m$ and $\{h_{\ell}\}_{\ell=1}^s$ be finite families of polynomials on the commutative ring $\mathbb{R}[x_1, \dots, x_n]$. Define $P\left(\left\{f_i\right\}_{i=1}^r \right)$, $M\left( \{g_j\}_{j=1}^m \right)$ and $I\left( \{h_{\ell}\}_{\ell=1}^s \right)$ as follows:
    $$
    \left\{ 
    \begin{aligned}
        &P\left(\left\{f_i\right\}_{i=1}^r \right) = \left\{ p + \sum_{i=1}^r q_i b_i:~p, q_i \in SOS_{\mathbb{R}[x_1, \dots, x_n]},~b_i \in M(\{f_i\}_{i=1}^r)  \right\} \\
        &M\left(\{g_j\}_{j=1}^m\right) = \left\{ \prod_{ i \in J } g_i:~J \subseteq \{1, \dots, m\} \right\} \\
        & I\left(\{h_\ell\}_{\ell=1}^s\right) = \left<h_1, \dots, h_s\right> = \left\{ \sum_{\ell = 1}^s  t_{\ell} h_{\ell} :~t_{\ell} \in \mathbb{R}[x_1, \dots, x_n] \right\},
    \end{aligned}
    \right.
    $$
    where we prescribe $\prod_{j \in \varnothing} g_i = 1$. Then the following statements are equivalent:
    \begin{enumerate}
        \item $\left\{ x \in \mathbb{R}^n:~f_i(x) \geq 0,~g_j(x) \neq 0,~h_{\ell}(x) = 0,~\forall i \in [r],~j \in [m],~\ell \in [s]  \right\} = \varnothing$.
        \item $\exists f \in P(\{f_i\}_{i=1}^r)$, $g \in M(\{g_j\}_{j=1}^m)$ and $h \in I(\{h_\ell\}_{\ell=1}^s)$, such that $f + g^2 + h = 0$.
    \end{enumerate}
\end{theorem}
We restrict the reformulation to a special case that $m = 0$, and the presented certificate degenerates to the Positivstellens\"atz (P-satz), \ie, $-1 = f + h$. Searching for sums of square terms can be efficiently done when all the involved polynomials admit a degree bound. Given a fixed degree bound $k \in \mathbb{N}$, let $[x]_k$ denote the vector of monomials on $\mathbb{R}[x_1, \dots, x_n]$ with degree bound $k$, \ie, 
$$
[x]_k = \begin{pmatrix}
    1, & x_1, & \cdots, & x_n, & x_1^2, &x_1 x_2, & \cdots, & x_n^2, &\cdots, & x_n^{k}
\end{pmatrix}^T \in \mathbb{N}^{n}_k. 
$$
It can be easily shown that the entries of $[x]_k$ are linearly independent over $\mathbb{F} \in \{\mathbb{R}, \mathbb{C}\}$. Note that any refutation $(f, h)$ take the following form
$$
f = \sum_{ I \subseteq [r] } p_I \prod_{i \in I} f_i,~p_I \in SOS_{\mathbb{R}[x_1, \dots, x_n]}; \quad h = \sum_{\ell = 1}^s t_{\ell} h_{\ell} ,~ t_{\ell} \in \mathbb{R}[x_1, \dots, x_n].
$$
Denote $\prod_{i \in I} f_i$ as $f_I$. If we restrict the search to be done on $\mathbb{N}^{n}_{2k}$, by \cite{PabloPhD} each SOS term $p_I$ is associated with an positive semidefinite matrix $P_I$ by $p_I = [x]_{d_I}^T P_I [x]_{d_I}$. Here the degree $d_I$ is chosen such that $\deg (p_I f_I) \leq 2k$, or $d_I = \lfloor k - \frac{1}{2} \deg(f_I)\rfloor$. We consider the reformulation of $f$ on the monomial subspace $\mathbb{N}_{2k}^n$, we denote $x^{\alpha} = x_1^{\alpha_1} x_2^{\alpha_2} \cdots x_n^{\alpha_n}$ for $\alpha \in \mathbb{N}^{n}_{2k}$, it holds that
$$
\begin{aligned}
    \sum_{\alpha \in \mathbb{N}_{2k}^{n} } f_{\alpha} x^{\alpha} = f &= \sum_{I \subseteq [r]} p_I f_I = \sum_{I \subseteq [r]} \left( [x]_{d_I}^T P_I [x]_{d_I} \right) f_I = \sum_{I \subseteq [r]} \left< [x]_{d_I} [x]_{d_I}^T f_I, P_I \right> \\
    &= \sum_{I \subseteq [r]} \left< \sum_{\alpha \in \mathbb{N}_{2k}^n} x^{\alpha} F_{I, \alpha}, P_I  \right> = \sum_{\alpha \in \mathbb{N}_{2k}^n} \sum_{I \subseteq [r]} \left<  F_{I, \alpha}, P_I  \right> x^{\alpha}.
\end{aligned}
$$
By the linear independence among the monomials, $f_{\alpha} = \sum_{I \subseteq [r]} \left< F_{I, \alpha}, P_I \right>$. As for the ideal membership, let $\texttt{coefvec}_{d}: \mathbb{R}[x_1, \dots, x_n] \to \mathbb{N}^{n}_{d}$ gives the coefficient vector of a polynomial with degree bound $d$ under the basis $[x]_{d}$. For each generator $h_{\ell}$, let $H_{\ell}^{2k}$ be the transformation such that 
$$
\texttt{coefvec}_{2k}(t_{\ell} h_{\ell} ) = H_{\ell}^{2k} \cdot \texttt{coefvec}_{\deg(t_{\ell})}(t_{\ell}), \quad \deg(t_{\ell}) = 2k - \deg(h_{\ell}).
$$
Then the SDP with degree bound $2k$ is presented as follows:
\begin{equation} \label{eqn:CstarSDP}
\textnormal{SDP}_k: \quad 
    \begin{array}{cl}
        \min\limits_{x_{\ell}, P_I} & 0 \\
         & -\delta_{\alpha, (0, \dots, 0)} = \sum_{\ell = 1}^s \left( H_{\ell}^{(2k)} \cdot x_{\ell}  \right)_{\alpha} + \sum_{I \subseteq [r]}\left< F_{I, \alpha}, P_I \right>, \quad \forall \alpha \in \mathbb{N}_{2k}^n; \\
         &x_{\ell} \in \mathbb{R}^{\binom{n + 2k - \deg h_{\ell}}{2k - \deg h_{\ell} }}, \quad P_I \in \mathbb{S}_{+}^{ \binom{n + d_I}{d_I} }.
    \end{array}
\end{equation}

\par On noncommutative polynomial rings and algebraically closed fields, the analog of the assertion in Theorem \ref{thm:psatz} might be more complicated. We refer readers to \cite{NCPsatz} for more detailed discussion. However, the results stated in Theorem \ref{thm:NCPsatzSyncGame} can be encoded in the hierarchical SDP (\ref{eqn:CstarSDP}) discussed above, by setting $\{f_i\}_{i=1}^r = \{1\}$, along with some slight modifications: the monomial basis $[x]_k$ should be replaced by a noncommutative version $[x]_{k}^{nc}$ over $\mathbb{C}^*\left<x_1, \dots, x_n\right>$, defined by
$$
[x]_k^{nc} =  \begin{pmatrix}
        1, & x_1, & \cdots, & x_n, & x_1^*, & \cdots, & x_n^*, & x_1^2, & \cdots, &{x_n^*}^k
    \end{pmatrix}^T,
$$
whence simplification can be applied for the game algebra $\mathcal{A}(\mathcal{G})$ since $x_{i, a}^* = x_{i, a}$ for any generator $x_{i, a}$. Thus, $[x]_{k}^{nc}$ can be replaced by an \textbf{involution-free} monomial basis, and we denote the corresponding subspace as $\mathbb{N}_{k}^{n}(nc)$. The SOS term is encoded by $([x]_{d_I}^{nc})^* P_I [x]_{d_I}^{nc} = \left< P_I, ({[x]_{d_I}^{nc}}^*)^T ([x]_{d_I}^{nc})^T \right>$ subject to $P_I \in \mathbb{S}_{+}^{\dim \mathbb{N}_{k}^{n}(nc)}$.
Each generator $h_{\ell}^*$ of the right ideal $\mathcal{I}(\mathcal{G})^* = \{ \sum_{\ell} h_{\ell}^* s_{\ell}:~s_{\ell} \in \mathcal{A}(\mathcal{G}) \}$ is associated with matrix ${H_{\ell}^{2k}}^*$ and a decision variable $x_{\ell}^*$ such that 
$$
\texttt{coefvec}_{2k}( h_{\ell}^* s_{\ell} ) = {H_{\ell}^{2k}}^* \cdot \texttt{coefvec}_{\deg(s_{\ell})}(s_{\ell}): \quad \deg(s_{\ell}) = 2k - \deg(h_{\ell}^*),
$$
which is analogous to the left ideal case. Further simplification can be made by applying the Gr\"obner basis method to contract the dimension of decision variables, as shown in the example in section 8.3.2. in \cite{NCPsatz}. Analogous to the commutative case, searching for a refutation in $\mathbb{N}_{k}^{n}(nc)$ takes polynomial time with respect to the size of the input/output space of $\mathcal{G}$ and the degree bound $k$.
\par We implemented the computation using Mathematica package NCAlgebra, and Python toolkit cvxpy and MOSEK. Codes are available in the GitHub repository\footnote{ \url{https://github.com/HeEntong/BCS/tree/main/mathematica_codes}.} in program files \texttt{PerfectSynGame.nb} and \texttt{NCPsatz\_SDP\_solver.py}.


\begin{problem}
    In light of the above computation using the first assertion in Theorem \ref{thm:NCPsatzSyncGame}, can we encode the second assertion concerning the existence of $qc$-strategy into the hierarchical SDP?
\end{problem}

\section{Hereditary synchronous subalgebra, and its equivalence to the $C^*$ subalgebra}

\label{sec:heredEqualsCstar}

Recall that a $*$-algebra $\mathcal{A}$ endowed with a SOS cone $SOS_{\mathcal{A}}$ possesses a partial ordering $\leq$, where $a \leq b$ if $b - a \in SOS_{\mathcal{A}}$ for self-adjoint $a, b \in \mathcal{A}$. Given a synchronous game $\mathcal{G} = (I, O, \lambda)$, a vector subspace $V \subseteq \mathbb{C}[\mathcal{F}(|I|, |O|)]$ is \textbf{hereditary} provided that any $f, h \in \mathbb{C}[\mathcal{F}(|I|, |O|)]$ subject to $0 \leq f \leq h$ and $h \in V$ implies $f \in V$. By Helton \etal's construction in \cite{helton2019algebras}, introducing the $*$-positive cone
$$
\mathcal{A}(\mathcal{G})^+ = SOS_{\mathbb{C}[\mathcal{F}(|I|, |O|)]} / \mathcal{I}(\mathcal{G}),
$$
and generalize the ordering by $a \leq b$ if $b - a \in \mathcal{A}(\mathcal{G})^+$. This makes $\mathcal{A}(\mathcal{G})$ into a \textbf{semi-pre-$C^*$-algeba} \cite{Ozawa2013}, and $\mathcal{A}(\mathcal{G})$ satisfies the Archimedean property since it is the quotient of a group algebra. We present the definition of the hereditary $*$-algebra of synchronous game $\mathcal{G}$, the perfect $hered$-strategy, and a reformulation of the perfect $C^*$-strategy.
\begin{definition}
    Given a synchronous game $\mathcal{G} = (I, O, \lambda)$, let $\mathcal{I}(\mathcal{G})$ be the $*$-closed, two-sided ideal with generating relations specified by (\ref{eqn:NonlocalGameRelations}). \cite{helton2019algebras} defined two subspaces $\mathcal{I}^h(\mathcal{G})$, $\mathcal{I}^c(\mathcal{G})$ containing $\mathcal{I}(\mathcal{G})$, as 
    \begin{equation}
    \mathcal{I}^h(\mathcal{G}) = \bigcap_{\substack{\mathcal{I}(\mathcal{G}) \subseteq V 
    \\ V \text{ is a hereditary subspace} 
    }} V; \quad \mathcal{I}^c(\mathcal{G}) = \bigcap_{ 
    \substack{
    \mathcal{I}(\mathcal{G}) \subseteq \ker \pi \\
    \mathbb{C}[\mathcal{F}(|I|, |O|)] \overset{\pi}{\longrightarrow} \mathcal{B}(\mathcal{H})
    }
    } \ker \pi.
    \end{equation}
    It is straightforward to verify that these two subspaces are ideals, which allows the definition of two subalgebras of $\mathcal{A}(\mathcal{G})$: \textbf{hereditary subalgebra} $\mathcal{A}^h(\mathcal{G}) = \mathbb{C}[\mathcal{F}(|I|, |O|)] / \mathcal{I}^h(\mathcal{G}) $, \textbf{$\bm{C^*}$ subalgebra} $\mathcal{A}^c(\mathcal{G}) = \mathbb{C}[\mathcal{F}(|I|, |O|)] / \mathcal{I}^c(\mathcal{G})$. We say that game $\mathcal{G}$ has a perfect $hered$-strategy if $\mathcal{A}^h(\mathcal{G}) \neq (0)$.
\end{definition}

\begin{remark}
\label{remark:CstarAndHered}
    Note that $\mathcal{G}$ has a perfect $C^*$-strategy if and only if $\mathcal{A}^c(\mathcal{G}) \neq (0)$. To explain why this is so, recall that the $C^*$-model is induced by a tuple $(\pi, \psi)$ where $\pi: \mathbb{C}[\mathcal{F}(|I|, |O|)] \to \mathcal{B}(\mathcal{H})$ is a unital $*$-homomorphism and $\psi \in \mathcal{H}$ is a state, by $\pp(a, b | i, j) = \psi^* \pi(x_{i, a}) \pi(x_{j, b}) \psi$. If $\mathcal{A}^c(\mathcal{G}) = (0)$, then by $\mathcal{A}^c(\mathcal{G}) \cong \pi\left(\mathbb{C}[\mathcal{F}(|I|, |O|)]\right)$, the only feasible unital $*$-homomorphism $\pi$ vanishing on $\mathcal{I}(\mathcal{G})$ would vanish on the whole free algebra. Thus either $\psi = 0$ or the underlying Hilbert space $\mathcal{H}$ admits dimension zero, then no perfect $C^*$-strategy exists. 
    For the converse, $\mathcal{A}^c(\mathcal{G}) \neq (0)$ indicates the existence of a unital $*$-homomorphism $\pi$ and a non-trivial state $\psi$ vanishing on $\mathcal{I}(\mathcal{G})$.
    \par Moreover, since $\ker \pi$ is itself a hereditary subspace, by the minimality of $\mathcal{I}^h(\mathcal{G})$, it follows that $\mathcal{I}^h(\mathcal{G}) \subseteq \mathcal{I}^c(\mathcal{G})$, or $\mathcal{A}^c(\mathcal{G}) \subseteq \mathcal{A}^h(\mathcal{G})$. We call $\mathcal{I}^h(\mathcal{G})$ the \textbf{hereditary closure} of $\mathcal{I}(\mathcal{G})$. Introducing the positive cone $\mathcal{A}^h(\mathcal{G})^+  = SOS_{\mathbb{C}[\mathcal{F}(|I|, |O|)]} / \mathcal{I}^h(\mathcal{G}) $ yields that $\mathcal{A}^h(\mathcal{G})$ is also a semi-pre-$C^*$-algebra.
\end{remark}

The hereditary closure of the $*$-closed ideal $\mathcal{I}(\mathcal{G})$ when $\mathcal{G}$ is a synchronous game is also $*$-closed. Thus, the $*$-closed assumption of $\mathcal{I}^h(\mathcal{G})$ is redundant, as stated in the definitions of \cite{helton2019algebras, harris2023universalitygraphhomomorphismgames}.

\begin{theorem}
    For a synchronous game $\mathcal{G}$, $\mathcal{I}^h(\mathcal{G})^* = \mathcal{I}^h(\mathcal{G})$.
\end{theorem}
\begin{proof}
    If suffices to show $\mathcal{I}(\mathcal{G}) \setminus SOS_{\mathbb{C}[\mathcal{F}(|I|, |O|)]} = \mathcal{I}^h(\mathcal{G}) \setminus SOS_{\mathbb{C}[\mathcal{F}(|I|, |O|)]}$. Note that for any hereditary subspace $V \subseteq \mathbb{C}[\mathcal{F}(|I|, |O|)]$, there is no $f, g \in V \cap SOS_{\mathbb{C}[\mathcal{F}(|I|, |O|)]}$ and $f + g \not \in V$. Thus, $\mathcal{I}^h(\mathcal{G})$ expands $\mathcal{I}(\mathcal{G})$ by including all hermitian SOS terms $f, g \in SOS_{\mathbb{C}[\mathcal{F}(|I|, |O|)]} \setminus \mathcal{I}(\mathcal{G})$ subject to $f + g \in \mathcal{I}(\mathcal{G})$. Moreover, any non-SOS term $a \in \left[ \mathcal{I}^h(\mathcal{G}) \setminus \mathcal{I}(\mathcal{G}) \right] \setminus SOS_{\mathbb{C}[\mathcal{F}(|I|, |O|)]}$ will violate the minimality of the hereditary closure, and is thus safe to be excluded from the construction. By the above assertions,
    \begin{equation}
    \begin{aligned}
    \mathcal{I}^h(\mathcal{G})^* 
    &= \left(\mathcal{I}^h(\mathcal{G}) \cap SOS_{\mathbb{C}[\mathcal{F}(|I|, |O|)]} \right)^* \sqcup \left(\mathcal{I}^h(\mathcal{G}) \setminus SOS_{\mathbb{C}[\mathcal{F}(|I|, |O|)]}\right)^* \\
    &= \left(\mathcal{I}^h(\mathcal{G}) \cap SOS_{\mathbb{C}[\mathcal{F}(|I|, |O|)]} \right)^* \sqcup \left(\mathcal{I}(\mathcal{G}) \setminus SOS_{\mathbb{C}[\mathcal{F}(|I|, |O|)]}\right)^* \\
    &= \left(\mathcal{I}^h(\mathcal{G}) \cap SOS_{\mathbb{C}[\mathcal{F}(|I|, |O|)]}\right) \sqcup \left(\mathcal{I}(\mathcal{G}) \setminus SOS_{\mathbb{C}[\mathcal{F}(|I|, |O|)]}\right) \\
    &= \left(\mathcal{I}^h(\mathcal{G}) \cap SOS_{\mathbb{C}[\mathcal{F}(|I|, |O|)]}\right) \sqcup \left(\mathcal{I}^h(\mathcal{G}) \setminus SOS_{\mathbb{C}[\mathcal{F}(|I|, |O|)]}\right) = \mathcal{I}^h(\mathcal{G}).
    \end{aligned}
    \end{equation}
    We conclude that $\mathcal{I}^h(\mathcal{G})$ is automatically a two-sided $*$-closed ideal.
\end{proof}

With all previous preliminaries, we are ready to show that the hereditary subalgebra of a synchronous game is indeed equivalent to the $C^*$ subalgebra. This assertion directly answers Problem 3.14 in \cite{helton2019algebras}, and moreover, it validates the inequality contested by Problem 3.16 in the same literature.

\begin{theorem}
    \label{thm:KakutaniSep} \cite{barvinok2002course}
    \textnormal{(\textbf{Eidelheri-Kakutani separation theorem})}. Let $V$ be a vector space and $A, B \subseteq V$ are disjoint non-empty convex subsets. Suppose $B$ admits a non-empty algebraic interior\footnote{The algebraic interior of $A$ is a subset $\textnormal{aint}_{V}(A) \subseteq A$ defined by $\textnormal{aint}_{V}(A) = \{a \in A:~\forall v \in V, \exists \delta_{v} > 0, \forall \delta \in [0, \delta_v], a + \delta v \in A \}$.}, then there exists a non-zero linear functional $\varphi: V \to \mathbb{R}$ such that 
    $$
    \forall a \in A, \quad \varphi(a) \leq \inf_{b \in B} \varphi(b).
    $$
\end{theorem}

\begin{lemma}
\label{lemma:ConvexCones}
    Both $\mathcal{A}^h(\mathcal{G})^+$ and $SOS_{\mathbb{C}[\mathcal{F}(|I|, |O|)]} \setminus \mathcal{I}^h(\mathcal{G})$ are convex cones.
\end{lemma}
\begin{proof}
    The first assertion follows by definition. As for the second assertion, consider $f, g \in SOS_{\mathbb{C}[\mathcal{F}(|I|, |O|)]} \setminus \mathcal{I}^h(\mathcal{G})$, assume that $f + g \in \mathcal{I}^h(\mathcal{G})$. By the property of the hereditary closure, since $0 \leq f, g \leq f + g$, it follows that $f, g \in \mathcal{I}^h(\mathcal{G})$, a contradiction. This implies $f + g \in SOS_{\mathbb{C}[\mathcal{F}(|I|, |O|)]} \setminus \mathcal{I}^h(\mathcal{G})$. It is not hard to show that $SOS_{\mathbb{C}[\mathcal{F}(|I|, |O|)]} \setminus \mathcal{I}^h(\mathcal{G})$ is closed under scalar multiplication in $\mathbb{R}_{>0}$, thus composing above two statements indicates that $SOS_{\mathbb{C}[\mathcal{F}(|I|, |O|)]} \setminus \mathcal{I}^h(\mathcal{G})$ is a (positive) convex cone.
\end{proof}

Note that Theorem \ref{thm:KakutaniSep} is a generalization of the Hahn-Banach separation theorem to an arbitrary vector space. With Lemma \ref{lemma:ConvexCones} we are ready to prove the main result in this section. The idea was proposed by Connor Paddock \cite{ConnorPersonal}, we give a formal proof here.

\begin{theorem}\label{thm:CstarEqualsHered}
    For any synchronous game $\mathcal{G}$, $\mathcal{A}^h(\mathcal{G}) = \mathcal{A}^c(\mathcal{G})$. 
\end{theorem}
\begin{proof}
    By Ozawa's remark in \cite{Ozawa2013}, the Archimedean property (or Combes axiom in his language) indicates the identity $1$ is an algebraic interior point of the $*$-positive cone $SOS_{\mathbb{C}[\mathcal{F}(|I|, |O|)]}$, and further of $SOS_{\mathbb{C}[\mathcal{F}(|I|, |O|)]} \setminus \mathcal{I}^h(\mathcal{G})$ when $\mathcal{A}^h(\mathcal{G})$ is non-trivial. By Theorem \ref{thm:KakutaniSep} and convexity condition in Lemma \ref{lemma:ConvexCones}, there exists a linear functional $\varphi: \Herm(\mathbb{C}[\mathcal{F}(|I|, |O|)]) \to \mathbb{C}$, satisfying 
    $$
    \varphi(\mathcal{I}^h(\mathcal{G})) \subseteq \mathbb{R}_{\leq 0}, \quad \varphi(SOS_{\mathbb{C}[\mathcal{F}(|I|, |O|)]} \setminus \mathcal{I}^h(\mathcal{G})) \subseteq \mathbb{R}_{\geq 0}.
    $$
    Without loss of generality, we choose $\varphi$ such that $\varphi(1) = 1$, which is consistent since $1 \in \varphi(SOS_{\mathbb{C}[\mathcal{F}(|I|, |O|)]} \setminus \mathcal{I}^h(\mathcal{G}))$. The image of $\varphi$ should vanish on $\mathcal{I}^h(\mathcal{G})$ since $\varphi$ is linear and $\mathcal{I}^h(\mathcal{G})$ is a $*$-closed subspace, thus $\varphi(\mathcal{I}^h(\mathcal{G})) = \{0\}$. As for hermitian element $h$, by $\Herm(\mathbb{C}[\mathcal{F}(|I|, |O|)]) = \{s - s':~s, s' \in SOS_{\mathbb{C}[\mathcal{F}(|I|, |O|)]} \}$ \cite{Ozawa2013}, we conclude that $\varphi(h) = \varphi(s) - \varphi(s')$ for some hermitian SOS terms $s, s'$. Thus, $\varphi(\Herm(\mathbb{C}[\mathcal{F}(|I|, |O|)])) \subseteq \mathbb{R}$. It is natural to extend $\varphi$ to the whole $*$-algebra, by defining the image of $a \in \mathbb{C}[\mathcal{F}(|I|, |O|)]$ as 
    $$
    \varphi(a) := \varphi\left(\frac{a + a^*}{2}\right) + i \varphi\left(\frac{a - a^*}{2i}\right).
    $$
    This shows that $\varphi(a^*) = \varphi(a)^*$ on $\mathbb{C}[\mathcal{F}(|I|, |O|)]$.
    \par To apply a variant of the Gelfand-Naimark-Segal (GNS) construction \cite{Arveson1976} on the $*$-algebra $\mathcal{A}^h(\mathcal{G}) = \mathbb{C}[\mathcal{F}(|I|, |O|)] / \mathcal{I}^h(\mathcal{G})$, we set $\varphi(a + \mathcal{I}^h(\mathcal{G})) := \varphi(a)$ for $a \in \mathbb{C}[\mathcal{F}(|I|, |O|)]$. Consequently, $\varphi$ separates $\{0\}$ and $\mathcal{A}^h(\mathcal{G})^+$. The definition is well-defined by the above discussion. We start by defining the sesquilinear form $[a, b] = \varphi(a^* b)$ for $a, b \in \mathcal{A}^h(\mathcal{G})$. The subset $\mathcal{N} = \{\ell \in \mathcal{A}^h(\mathcal{G}):~\varphi(\ell^* \ell) = 0\}$ is a linear subspace. Moreover, by the Cauchy-Schwarz inequality, for any $a \in \mathcal{A}^h(\mathcal{G}), \ell \in \mathcal{N}$, due to inequality
    $$
    0 \leq |\varphi((a\ell)^* a\ell)|^2 \leq \varphi(\ell^* \ell) \varphi((a^* a\ell)^* a^*a\ell) = 0,
    $$
    we have $a \mathcal{N} = \mathcal{N}$, that is, $\mathcal{N}$ is a left ideal. As $1 \not \in \mathcal{N}$, we can define the non-trivial quotient $\widetilde{\mathcal{H}} = \mathcal{A}^h(\mathcal{G}) / \mathcal{N}$ equipped with the inner product $\left<\cdot, \cdot\right>$, where
    $$
    \forall a, b \in \mathcal{A}^h(\mathcal{G}), \quad \left<a + \mathcal{N}, b + \mathcal{N}\right> := [a, b] = \varphi(a^*b),
    $$
    we obtain a pre-Hilbert space $\widetilde{\mathcal{H}}$, which can be routinely completed into a Hilbert space $\mathcal{H}$ \cite{Arveson1976}. A $*$-representation can be induced on $\mathcal{H}$ by setting $\pi(a)(x + \mathcal{N}) = ax + \mathcal{N}$ for any $a \in \mathcal{A}^h(\mathcal{G}), x + \mathcal{N} \in \mathcal{A}^h(\mathcal{G}) / \mathcal{N}$. The image $\pi(\mathcal{A}^h(\mathcal{G}) / \mathcal{N})$ is indeed bounded, since for any $a \in \mathcal{A}^h(\mathcal{G})$, by the Archimedean property, there exists $\lambda_a \in \mathbb{R}_{>0}$ and $s_a \in \mathcal{A}^h(\mathcal{G})^+$, such that $\lambda_a - a^* a = s_a$, then
    $$
    \begin{aligned}
    \left< \pi(a)(x + \mathcal{N}), \pi(a)(x + \mathcal{N}) \right> &= \left<ax + \mathcal{N}, ax + \mathcal{N}\right> = \varphi(x^* a^* a x) \\
    &= \varphi(x^*(\lambda_a - s_a) x) \leq \lambda_a \cdot \varphi(x^* x).
    \end{aligned}
    $$
    Again by the Archimedean property, $\varphi(x^* x)$ is bounded, indicating that $ = \|\pi(a)\| < \infty$. For any $a \in \mathcal{I}^h(\mathcal{G})$, $\left<\pi(a)(x + \mathcal{N}), \pi(a)(x + \mathcal{N})\right>$ evaluates to $\varphi((ax)^* (ax)) = 0$, by $ax \in \mathcal{I}^h(\mathcal{G})$ for any $x \in \mathcal{A}^h(\mathcal{G})$ and $\mathcal{I}^h(\mathcal{G})^* \mathcal{I}^h(\mathcal{G}) \subseteq \mathcal{I}^h(\mathcal{G})$. Thus $\pi(\mathcal{I}^h(\mathcal{G}) = \{0\}$.
    We end up with a valid $*$-homomorphism $\pi: \mathbb{C}[\mathcal{F}(|I|, |O|)] \to \mathcal{B}(\mathcal{H})$ vanishing on $\mathcal{I}^h(\mathcal{G})$, such that $\mathcal{I}(\mathcal{G}) \subseteq \mathcal{I}^h(\mathcal{G}) \subseteq \ker \pi$.
    \par Note that $\mathcal{A}^c(\mathcal{G})$ is the maximal subalgebra of $\mathcal{A}(\mathcal{G})$ such that a $*$-homomorphism $\pi: \mathbb{C}[\mathcal{F}(|I|, |O|)] \to \mathcal{B}(\mathcal{H})$ subject to $\mathcal{I}(\mathcal{G}) \subseteq \ker \pi$ exists. Due to the presence of such map $\pi$ on any non-trivial subalgebra $\mathcal{A}^h(\mathcal{G})$, by the maximality of $\mathcal{A}^c(\mathcal{G})$, we claim that $\mathcal{A}^h (\mathcal{G}) \subseteq \mathcal{A}^c(\mathcal{G})$. Combined with the reverse inclusion in Remark \ref{remark:CstarAndHered}, we have $\mathcal{A}^h(\mathcal{G}) = \mathcal{A}^c(\mathcal{G})$, as desired.
\end{proof}

\begin{proposition}
    $\omega_{hered}(G) \leq \vartheta(\overline{G}) \leq \chi_{hered}(G)$.
\end{proposition}

\begin{proof}
    As a direct consequence of Theorem \ref{thm:CstarEqualsHered}, $\omega_{hered}(G) = \omega_{C^*}(G)$ and $\chi_{hered}(G) = \chi_{C^*}(G)$. Combining the inequality $\omega_{C^*}(G) \leq \vartheta(\overline{G}) \leq \chi_{C^*}(G)$ the result follows.
\end{proof}

\section{Quantum version of Boolean constraint systems and NP-hardness reductions}
We now turn to explore the quantum version of general Boolean constraint systems (BCS). BCS arises in the study of the two-player non-local games with shared quantum states and measurement operators \cite{Cleve2014}. The transformation of BCS from the classical setting to its quantum counterpart is natural, as follows:

\begin{definition}
    A BCS is described by a tuple $\mathscr{B} = (X, \{U_i, p_i \}_{i=1}^m )$, where $X = \{x_j\}_{j=1}^n$ is the \textbf{universe} of Boolean variables taking values in $\{0, 1\}$. For each $i \in [m]$, $U_i \subseteq X$ is the \textbf{context} of constraint $i$, and $p_i: \{0,1\}^{|U_i|} \to \{0, 1\}$ is a polynomial associated with $U_i$, which evaluates to $1$ if and only if the Boolean assignment of $U_i$ is satisfying. $\mathscr{B}$ has a classical satisfying assignment if there exists a Boolean value assignment of $X$ such that all constraints $p_i: i \in [m]$ evaluate to $1$.
    \par Following Ji \cite{Ji2013}, $\mathscr{B}$ has \textbf{quantum (operator) satisfying assignment} if there exists a Hilbert space $\mathcal{H}$ with $\dim \mathcal{H} < \infty$, and a unital $*$-homomorphism $\pi: X \to \mathcal{B}(\mathcal{H})$ satisfying $p_i(\pi(x):~x \in U_i) = I_{\mathcal{H}}$. Moreover, it is required that operators in the same context pairwise commute, \ie, for any $i \in [m]$, $x, y \in U_i$, $[\pi(x), \pi(y)] = 0$. Such satisfiability is notably equivalent to the perfect $q$-strategy defined in Section \ref{sec:AlgebraDefinition}.
\end{definition}

\begin{remark}
    For a certain class of BCS problems $\textnormal{\texttt{BCS}}$, we use $\textnormal{\texttt{BCS}}^*$ to denote its quantum version. In the latter context, we use $\textnormal{\texttt{A}} \leq_p \textnormal{\texttt{B}}$ to specify that the problem class $\textnormal{\texttt{A}}$ is \textbf{polynomial-time (or Karp) reducible} to class $\textnormal{\texttt{B}}$, which indicates that solving $\textnormal{\texttt{A}}$ is not harder than solving $\textnormal{\texttt{B}}$.
\end{remark}

\begin{theorem} 
\label{thm:reductions}
\cite{Ji2013}
    $\prob{3-SAT}^* \leq_p \prob{3-Coloring}^*$, $ \prob{3-SAT}^* \leq_p \prob{1-in-3-SAT}^*$, $\prob{k-SAT}^* \leq_p  \prob{3-SAT}^*$, and $\prob{3-SAT} \leq_p \prob{3-SAT}^*$.
\end{theorem}

Among the many propositions of Theorem \ref{thm:reductions}, we are particularly concerned with the first one. Ji's reduction utilizes the classical reduction gadget graph with extra structures that ensure the local commutativity of the operator assignment. More specifically, it is necessary to add triangular prisms between vertices corresponding to variables in the same clause (same context) to force commutativity in the gadget. Harris exploits this prism in \cite{harris2023universalitygraphhomomorphismgames} to deduce a so-called weak $*$-equivalence between a synchronous game $\mathcal{G}$ and a $3$-coloring game on its induced graph. Ji's remark right before Lemma 3 in the literature suggests that such construction is rather specific and does not apply to other quantum versions of NP-hardness reductions. In light of Ji's 
comment, we provide another reduction that preserves the quantum satisfiability criteria, which is the main result of this section.

\begin{theorem}
\label{thm:3satToClique}
    $\prob{3-SAT}^* \leq_p \prob{Clique}^*$.
\end{theorem}

\begin{proof}
    Given any $\prob{3-SAT}^*$ instance $\phi = \bigwedge_{i=1}^m C_{i}$, where $C_i = \bigvee_{\alpha = 1}^{3} x_{i, \alpha}$. We denote the support of underlying variables as $x$, and each literal satisfies $x_{i, \alpha} \in x \cup \neg x$. For conciseness, we assume that within each clause there is no $x_{i, \alpha} = x_{i, \alpha'}$ or $x_{i, \alpha} = \neg x_{i, \alpha'}$ for $\alpha \neq \alpha' \in [3]$.
    \par The undirected gadget graph $G_{\phi}$ is built from $\phi$ as follows: for each clause $C_i$, create a cluster of vertices $S_i = \{v_{i, \alpha}\}_{\alpha=1}^3$ (\ie, $v_{i, \alpha}$ corresponds to literal $x_{i, \alpha}$), and set $V(G_{\phi}) = \bigsqcup_{i=1}^m S_i$. The edge set $E(G_{\phi}) = \left\{ \{ v_{i, \alpha}, v_{j, \alpha'} \}:~\forall i \neq j \in [m],~\alpha, \alpha' \in [3],~x_{i, \alpha} \neq \neg x_{j, \alpha'}     \right\}$. Classically, $\phi$ is satisfiable if and only if the gadget graph $G_{\phi}$ admits a $m$-clique.
    \par Like in \cite{Ji2013}, we need to modify the classical gadget accordingly to validate the reduction in the quantum framework. For each vertex pair $(v_{i, \alpha}, v_{j, \alpha'}) \in V(G_{\phi}) \times V(G_{\phi})$, we delete edges in $G_{\phi}$ to ensure that \textbf{(i)} if $x_{i, \alpha} = x_{j, \alpha'}$, then $N_{v_{i, \alpha}} \cup \{v_{i, \alpha}\} = N_{v_{j, \alpha'}} \cup \{v_{j, \alpha'}\}$; \textbf{(ii)} if $x_{i, \alpha} = \neg x_{j, \alpha'}$, then $N_{v_{i, \alpha}} \cup \{v_{i, \alpha}\} = N_{S_j \setminus \{ v_{j, \alpha'} \} } \cup ( N_{v_{i, \alpha}} \cap S_{j} \setminus \{v_{j, \alpha'}\} )$ (or equivalently, $N_{v_{j, \alpha'}} \cup \{v_{j, \alpha'}\} = N_{S_i \setminus \{ v_{i, \alpha} \} } \cup ( N_{v_{j, \alpha'}} \cap S_i \setminus \{v_{i, \alpha}\}  )$). This operation is applied recursively until there is no further graph contraction. We denote the resulting graph as $G_{\phi}^*$, \ie, the quantum gadget.
    \par If the $\prob{3-SAT}^*$ instance $\phi$ admits a satisfying idempotent operator assignment $\pi: x \cup \neg x \to \mathcal{B}(\mathcal{H})$ which sends $\neg x_{t}$ to $I - \pi(x_t)$ where $x_t$ is the positive variable in $x$, then at each clause $C_i$, $\prod_{\alpha \in [3]} (I - \pi(x_{i, \alpha})) = 0$. By the locally commutativity, this is equivalent to $\sum_{\alpha \in [3]} \pi(x_{i, \alpha}) \succeq I$. Without loss of generality, there exists a compression of projections $\iota \in \mathcal{B}(\mathcal{H})$ subject to $0 \preceq \iota \preceq I$ on $\pi(x \cup \neg x)$, such that $\image\left\{\iota(\pi(x_{i, \alpha}))\right\} \subseteq \image\left\{ \pi(x_{i, \alpha})  \right\}$, and $\sum_{\alpha \in [3]} \iota(\pi(x_{i, \alpha})) = I$. Define the compression map $\tau = \iota \circ \pi$. The quantum strategy of $\Hom(K_m, G_{\phi}^*)$ is constructed as follows: assign $X_{i, v_{i, \alpha}} = \tau(x_{i, \alpha})$, and $X_{i, v_{j, \alpha'}} = 0$ for any $\alpha' \in [3]$ and $j \neq i$. Such construction is valid: whenever $v_{i, \alpha}$ and $v_{j, \alpha'}$ are disconnected in $G_{\phi}^*$, there are following cases \textbf{(i)} if $i = j$ and $\alpha \neq \alpha'$, then by $\sum_{\alpha \in [3]} \tau(x_{i, \alpha}) = I$ and operators are idempotents, equality $X_{i, v_{i, \alpha}} X_{i, v_{i, \alpha'}} = 0$ is automatic; \textbf{(ii)} if $i \neq j$ and $x_{i, \alpha} = \neg x_{j, \alpha'}$ then by $\image(\tau(x_{i, \alpha}) \tau(x_{j, \alpha'})) \subseteq \image(\pi(x_{i, \alpha}) \pi(x_{j, \alpha'})) = \image(\pi(x_{i, \alpha}) - \pi(x_i, \alpha)^2) = \{0\}$, and thus $X_{i, \alpha} X_{j, \alpha'} = 0$; \textbf{(iii)} there exists $v_{k, \alpha''}$ such that $x_{i, \alpha} = x_{k, \alpha''}$, and $v_{i, \alpha} \sim_{G_{\phi}} v_{k, \alpha''}$ but $v_{j, \alpha'} \not \sim_{G_{\phi}} v_{k, \alpha''}$. Then it holds that $j = k$ and $\image(\tau(x_{i, \alpha}) \tau(x_{j, \alpha'})) = \image(\tau(x_{k, \alpha''}) \tau(x_{j, \alpha'})) = \{0\}$, implying $X_{i, v_{i, \alpha}} X_{j, v_{j, \alpha'}} = 0$; \textbf{(iv)} $i \neq j$ and there exists $v_{k, \alpha''}$ such that $x_{i, \alpha} = \neg x_{k, \alpha''}$, and $v_{i, \alpha} \in N_{S_k \setminus \{v_{k, \alpha''}\}}$ but $v_{j, \alpha'} \not \in N_{S_k \setminus \{v_{k, \alpha''}\}}$ in $G_{\phi}$. In $G_{\phi}^*$, we have $\tau\left( \sum_{\beta \neq \alpha'' } x_{k, \beta} \right) = I - \tau(x_{k, \alpha''})$, and by $\pi(x_{i, \alpha}) = I - \pi(x_{k, \alpha''})$ we have $\image(\tau(x_{i, \alpha})) \subseteq \image(I - \pi(x_{k, \alpha''}))$. Since in $G_{\phi}$, $v_{j, \alpha'} \not \in N_{S_k \setminus \{v_{k, \alpha''}\}}$, we have $ \pi\left( \sum_{\beta \neq \alpha'' } x_{k, \beta} \right) \pi(x_{j, \alpha}) = 0$ (due to an analogous argument to \textbf{(i)} and \textbf{(ii)}, in our construction, operators assigned to non-adjacent vertices in $G_{\phi}$ always yield product zero). Then in $G_{\phi}^*$, it holds that $\image( \tau(x_{i, \alpha}) \tau(x_{j, \alpha'})  ) \subseteq \image (\pi(x_{i, \alpha}) \tau(x_{j, \alpha'}) ) = \image( (I - \pi(x_{k, \alpha''})) \tau(x_{j, \alpha'})  ) \subseteq \image ( (I - \tau(x_{k, \alpha''})) \tau(x_{j, \alpha'}) ) = \image\left( \tau\left( \sum_{\beta \neq \alpha'' } x_{k, \beta} \right) \tau(x_{j, \alpha'}) \right) \subseteq \image\left( \pi\left( \sum_{\beta \neq \alpha'' } x_{k, \beta} \right) \pi(x_{j, \alpha})\right) = \{0\}$, and again we have $X_{i, v_{i, \alpha}} X_{j, v_{j, \alpha'}} = 0$.
    \par By the above construction, the synchronous condition is fulfilled by $X_{i, v_{i, \alpha}} X_{i, v_{j, \alpha'}} = \delta_{i, j} \delta_{\alpha, \alpha'} X_{i, v_{i, \alpha}}$, and the invalid responses vanish since whenever $v_{i, \alpha} \not \sim_{G_{\phi}^*} v_{j, \alpha'}$ with $i \neq j$, $X_{i', v_{i, \alpha}} X_{j', v_{j, \alpha'}} = \delta_{i, i'} \delta_{j, j'} X_{i, v_{i, \alpha}} X_{j, v_{j, \alpha'}} = 0$ for any $i' \neq j'$. Finally, for any $i \in [m]$, $\sum_{v \in V(G_{\phi}^*)} X_{i, v} = \sum_{v \in S_i} X_{i, v} = I$, and hence $\{X_{i, v}\}_{i \in [m], v \in V(G_{\phi}^*)}$ is a perfect quantum strategy of the graph homomorphism game $\Hom(K_m, G_{\phi}^*)$.
    \par If the $\prob{Clique}^*$ instance $G_{\phi}^*$ admits a perfect quantum strategy $\{X_{i, v}\}_{i \in [m], v \in V(G_{\phi}^*)}$, we assign operator $X_v := \sum_{i \in [m]} X_{i, v}$ to the literal corresponds to each $v \in V(G_{\phi}^*)$. By Lemma \ref{lemma:equivelnceWithSameNeighborhood} the structure of $G_{\phi}^*$ ensures that \textbf{(i)} the corresponding literals with the same underlying variable will be assigned the same operator; \textbf{(ii)} two literals which are mutually negated underlying variables will be assigned operators that sum up to $I$, since if $x_{i, \alpha} = \neg x_{j, \alpha'}$ then $X_{v_{i, \alpha}} = \sum_{v \in S_j} X_{v} - X_{v_{j, \alpha'}} = I - X_{v_{j, \alpha'}}$ by Lemma \ref{lemma:sumToIdentityCluster}.
    

     \par Since $X_v^2 = \sum_{i, j \in [m]} X_{i, v} X_{j, v} = \sum_{i \in [m]} X_{i, v}^2 = \sum_{i \in [m]} X_{i, v} = X_{v}$, variables $x$ are assigned to idempotents, and each clause $C_i$ is satisfied since $\sum_{v \in S_i} X_v = I$. Locally commutativity is naturally deduced, by $u \not \sim_{G_{\phi}^*} v$ for any $u \neq v \in S_i$, $[X_{u}, X_{v}] = \sum_{i, j \in [m]} [X_{i, u} X_{j, v} - X_{j, v} X_{i, u}] = 0$. Therefore, we obtain a valid quantum satisfying assignment for $\phi$ from $\{X_{i, v}\}_{i \in [m], v \in V(G_{\phi}^*)}$.
\end{proof}

\begin{lemma}\footnote{This property was initially verified through computer-assisted computation of Gr\"obner basis, for details see \url{https://github.com/HeEntong/BCS/blob/main/mathematica_codes/3SAT_Clique_Gadget.nb}.}
    \label{lemma:sumToIdentityCluster}
    For every $i \in [m]$, $\sum_{v \in S_i} \sum_{j \in [m]} X_{j, v} = I$.
    \end{lemma}
    \begin{proof}
        We present an auxiliary claim: suppose $\mathcal{V}_t = \left\{ (v_{1}, \dots, v_m):~\forall \ell \in [m],~v_\ell \not \in S_t \right\}$ is the set of vertex sets such that none of its members is from the vertex cluster $S_t$, then it holds that 
        \begin{equation}
        \label{ass:auxToReduction}
        \sum_{(v_i)_{i=1}^m \in \mathcal{V}_t} x_{1, v_1} x_{2, v_2} \cdots x_{m, v_m} = \sum_{\substack{(v_i)_{i=1}^m \in \mathcal{V}_t
        \\ v_i \sim_{G_{\phi}^*} v_{i+1} 
        }} x_{1, v_1} x_{2, v_2} \cdots x_{m, v_m} = 0.
        \end{equation}
        By the pigeonhole principle, since none of $(v_1, \dots, v_m)$ is from $S_t$, there must be a pair of vertices $v_i, v_j$ where $i, j \in [m]$, such that $v_i, v_j$ falls into the same cluster $S_k$ where $k \in [m] \setminus \{t\}$, with either $v_i = v_j$ or $v_i \neq v_j$. There might be multiple such pairs, and without loss of generality, it suffices to examine those connected by the shortest vertex disjoint path $v_i \leadsto v_j$ (we herein assume $i < j$), and if there are several candidates, we focus our analysis on a single selected pair. According to the length of such shortest paths, we decompose $\mathcal{V}_t$ as follows: $\mathcal{V}_t = \bigsqcup_{\substack{n \geq 2 \\ \exists k \in [m] \setminus \{t\} }} \left\{ (v_1, \dots, v_m):~v_i \leadsto v_j \text{ a shortest path},~v_i, v_j \in S_k,~j - i = n  \right\} $. If we fix any $v_i, v_j$ that mark the start and end of the shortest path,
        \begin{align*}
        \begin{aligned}
            \sum_{\substack{(v_\ell)_{\ell=1}^m \in \mathcal{V}_t: \\ v_i, v_j \in S_k
            , v_i \leadsto v_j} } x_{1, v_1} x_{2, v_2} \cdots x_{m, v_m} &= \sum_{\substack{(v_\ell)_{\ell=1}^m \in \mathcal{V}_t:  \\ v_i, v_j \in S_k, v_i \leadsto v_j
            } } \prod_{\ell < i} x_{\ell, v_\ell} \cdot ( x_{i, v_i} \cdots x_{j, v_{j}}) \cdot \prod_{\ell > j} x_{\ell, v_\ell} \\
            &= \sum_{\substack{(v_\ell)_{\ell < i \lor \ell > j}:  \\ (v_{\ell})_{\ell=1}^m \in \mathcal{V}_t
            } } \prod_{\ell < i} x_{\ell, v_\ell}  x_{i, v_i} \prod_{\ell=i+1}^{j-1} \left[\sum_{w \in V(G_{\phi}^*)} x_{\ell, w}\right]   x_{j, v_j} \prod_{\ell > j} x_{\ell, v_\ell} \\
            &= \prod_{\ell < i} x_{\ell, v_\ell} \left( x_{i, v_i} x_{j, v_j} \right) \prod_{\ell > j} x_{\ell, v_\ell} = 0.
        \end{aligned}
        \end{align*}
        The proof of the claim in (\ref{ass:auxToReduction}) is done by noting that $\sum_{(v_i)_{i=1}^m \in \mathcal{V}_t} x_{1, v_1} x_{2, v_2} \cdots x_{m, v_m} = \sum_{n \geq 2} \sum_{\substack{v_i \leadsto v_j \text{ shortest} \\v_i, v_j \in S_k, k \in [m]\setminus \{t\}, j - i = n}} x_{1, v_1} x_{2, v_2} \cdots x_{m, v_m} = 0$. Utilizing this identity, we can deduce the result by the following equality: fix any $i \in [m]$,
        \begin{equation}
        \begin{aligned}
        \prod_{j=1}^m \sum_{v \in V(G_{\phi}^*)} x_{j, v} &- \sum_{j=1}^{m} \left(\prod_{\ell < j} \sum_{w \in V(G_{\phi}^*)} x_{\ell, w}\right) \cdot \left( \sum_{v \in S_i} x_{j, v}  \right) \cdot \left( \prod_{\ell > j} \sum_{w \in V(G_{\phi}^*)} x_{j, w} \right) \\
        &= \sum_{\substack{(v_\ell)_{\ell=1}^m: \\
        \forall j \in [m],~v_j \not \in S_i
        }} x_{1, v_1} x_{2, v_2} \cdots x_{m, v_m} = \sum_{(v_\ell)_{\ell=1}^m \in \mathcal{V}_i} x_{1, v_1} x_{2, v_2} \cdots x_{m, v_m} = 0,
        \end{aligned}
        \end{equation}
        which is equivalent to $1 = \sum_{j=1}^m \sum_{v \in S_i} x_{j, v}$. To obtain the assertion stated in the lemma, we take the unital $*$-homomorphism $\pi: \mathcal{A}(\Hom(K_m, G_{\phi}^*)) \to \mathcal{B}(\mathcal{H})$ with $x_{i, v} \mapsto X_{i, v}$ and $1 \mapsto I$.
     \end{proof}

\begin{lemma}
     \label{lemma:equivelnceWithSameNeighborhood}
            If $N_u \cup \{u\} = N_v \cup \{v\}$ holds for two vertices $u \sim_{G_{\phi}^*} v $, $\{x_{i, v}\}_{i \in [m], v \in V(G_{\phi}^*)}$ are the generators of $\mathcal{A}(\Hom(K_m, G_{\phi}^*))$ then $\sum_{i \in [m]} x_{i, u} = \sum_{i \in [m]} x_{i, v}$. More generally, given vertex subsets $\mathscr{I} \subseteq S_i$ and $\mathscr{J} \subseteq S_j$ with $i \neq j$ subject to $\mathscr{I} \cap N_{\mathscr{J}} \neq \varnothing$ (or equivalently, $\mathscr{J} \cap N_{\mathscr{I}} \neq \varnothing$), and $N_{\mathscr{I}} \cup ( \mathscr{I} \cap N_{\mathscr{J}}  ) = N_{\mathscr{J}} \cup ( \mathscr{J} \cap N_{\mathscr{I}} )$, we have $\sum_{u \in \mathscr{I}} \sum_{i \in [m]} x_{i, u} = \sum_{v \in \mathscr{J}} \sum_{j \in [m]} x_{j, v} $.
     \end{lemma}
     \begin{proof}
         We start by the following statement: suppose $\mathcal{I} = \left< h_1, \dots, h_\ell \right>$ is a $*$-closed two-sided ideal of a $*$-algebra $\mathcal{A}$ where $h_{\ell} \in \mathcal{A}$, and there exists a $*$-automorphism $\gamma$ on $\mathcal{A}$ such that $\left\{ h_1, \dots, h_{\ell} \right\} = \left\{ \gamma(h_1), \dots, \gamma(h_\ell)  \right\}$, then for any $f \in \mathcal{I}$, we have $\gamma(f) \in \mathcal{I}$. The proof is straightforward: $f \in \mathcal{I}$ implies that there exists $\{p_i \}_{i=1}^{\ell}, \{ q_i \}_{i=1}^{\ell} \subseteq \mathcal{A}$ such that $f = \sum_{i=1}^{\ell} p_i h_i q_i$. Then $\gamma(f) = \sum_{i=1}^{\ell} \gamma(p_i) \gamma(h_i) \gamma( q_i) \in \mathcal{I}$ since each $\gamma(h_i)$ is in $\{h_1, \dots, h_\ell\}$. 
         \par It can be checked if $u \sim_{G_{\phi}^*} v$ (thus $u, v$ are in distinct clusters) with $N_u \cup \{u\} = N_v \cup \{v\}$, then the generating relations of $\mathcal{I}(\Hom(K_m, G_{\phi}^*))$ is invariant under the unital $*$-automorphism $\gamma$ such that $\gamma(x_{i, u}) = x_{i, v}$, $\gamma(x_{i, v}) = x_{i, u}$ for all $i \in [m]$, and $\gamma|_{\{X_{i, w}:~w \neq u, v\}}$ is the identity map. Suppose $u$ is in cluster $S$, then $\sum_{w \in S} \sum_{i \in [m]} x_{i, w} = 1$, by Lemma \ref{lemma:sumToIdentityCluster} we have $\sum_{w \in S} \sum_{i \in [m]} \gamma(x_{i, w}) = \sum_{i \in [m]}  \gamma(x_{i, u}) + \sum_{w \in S \setminus \{u\} } \sum_{i \in [m]} \gamma(x_{i, w}) = \sum_{i \in [m]} x_{i, v} + \sum_{w \in S \setminus \{u\} } \sum_{i \in [m]} x_{i, w} = 1 = \sum_{i \in [m]} x_{i, u} + \sum_{w \in S \setminus \{u\} } \sum_{i \in [m]} x_{i, w}$, indicating that $\sum_{i \in [m]} x_{i, u} = \sum_{i \in [m]} x_{i, v}$. The generalization to the multiple vertices case is done by contracting $\mathscr{I}$ and $\mathscr{J}$ into two single vertices $v_{\mathscr{I}}$ and $v_{\mathscr{J}}$ respectively, and assign $x_{i, v_{\mathscr{I}}} \gets \sum_{u \in \mathscr{I}} x_{i, u}$, $x_{i, v_{\mathscr{J}}} \gets \sum_{u \in \mathscr{J}} x_{i, u}$.By an analogous argument one can show that $\sum_{i \in [m]} x_{i, v_{\mathscr{I}}} = \sum_{i \in [m]} x_{i, v_{\mathscr{J}}}$, as desired.
     \end{proof}

\begin{remark}
    The derivation of Lemma \ref{lemma:sumToIdentityCluster} replies purely on the algebraic properties of members of $\mathcal{A}(\Hom(K_m, G_{\phi}^*))$. Given that $X_{i, v}$ are positive operators, the product $\prod_{\ell=1}^m X_{\ell, v_{\ell}}$ with $v_i, v_j \in S_k$ is annihilated by noting 
    $$
    \begin{aligned}
    \Tr\left( X_{i, v_i} X_{i+1, v_{i+1}} \cdots X_{j, v_j}  \right) =  \Tr\left(X_{i+1, v_{i+1}} \cdots X_{j, v_j} X_{i, v_i}  \right) = 0 
    \implies \prod_{\ell=i}^{j} X_{\ell, v_{\ell}} \bigg|_{v_i, v_j \in S_k} = 0.
    \end{aligned}
    $$
     Consequently we have $\prod_{\ell=1}^m X_{\ell, v_{\ell}} = \prod_{\ell < i} X_{\ell, v_{\ell}} \cdot \prod_{\ell=i}^{j} X_{\ell, v_{\ell}} \cdot \prod_{\ell > j} X_{\ell, v_{\ell}} = 0$, which directly implies that the sum $\sum_{(v_{\ell})_{\ell=1}^m \in \mathcal{V}_{i}}  \prod_{\ell=1}^m X_{\ell, v_{\ell}}$ vanishes. This provides a more concise proof of the auxiliary claim.
\end{remark}
It is not hard to see the gadget utilized in the proof of Theorem \ref{thm:3satToClique} can also be applied to prove the polynomial reduction from $\prob{SAT}^*$ to $\prob{Clique}^*$, which is a more general analog. The structure of the reduction remains almost unchanged, and we state the result without proof here.

\begin{proposition}
    \label{prop:satToClique}
    $\prob{SAT}^* \leq_p \prob{Clique}^*$.
\end{proposition}

\section*{Acknowledgements}
I thank Connor Paddock for his valuable introduction to non-local games and insightful discussions, which inspired the major topics of this report. I thank Prof. Anne Broadbent and Prof. Chenshu Wu for being my supervisor and co-supervisor during the internship. This work is jointly supported by Mitacs and The University of Hong Kong through the Globalink Research Internship 2024.


\bibliography{refs} \bibliographystyle{alpha}

\end{document}